%% file: main.tex
\documentclass[]{lmcs} 
\pdfoutput=1

\usepackage{lastpage}
\lmcsdoi{18}{3}{27}
\lmcsheading{}{\pageref{LastPage}}{}{}%
{Aug.~04,~2021}{Aug.~31,~2022}{}


\usepackage{amssymb}
\usepackage{mathpartir}
\usepackage{stmaryrd}
\usepackage{xspace}
\usepackage{prooftree}
\usepackage[utf8]{inputenc}
\usepackage{agda}
\usepackage{newunicodechar}
\newunicodechar{∞}{\ensuremath\infty}
\newunicodechar{ℕ}{\ensuremath\Nat}
\newunicodechar{𝕍}{\ensuremath\Message}
\newunicodechar{∀}{\ensuremath\forall}
\newunicodechar{α}{\ensuremath\action}
\newunicodechar{ε}{\ensuremath\varepsilon}
\newunicodechar{λ}{\ensuremath\lambda}
\newunicodechar{φ}{\ensuremath\actions}
\newunicodechar{∈}{\ensuremath\in}
\newunicodechar{∃}{\ensuremath\exists}
\newunicodechar{∷}{\ensuremath{::}}
\newunicodechar{∪}{\ensuremath\cup}
\newunicodechar{∩}{\ensuremath\cap}
\newunicodechar{⊆}{\ensuremath\subseteq}
\newunicodechar{₀}{\ensuremath{_0}}
\newunicodechar{₁}{\ensuremath{_1}}
\newunicodechar{₂}{\ensuremath{_2}}
\newunicodechar{₃}{\ensuremath{_3}}
\newunicodechar{⊢}{\ensuremath\vdash}
\newunicodechar{⊣}{\ensuremath\dashv}
\newunicodechar{⊥}{\ensuremath\bot}
\newunicodechar{ᵢ}{\ensuremath{_i}}
\newunicodechar{Σ}{\ensuremath\Sigma}
\newunicodechar{⟦}{\ensuremath\llbracket}
\newunicodechar{⟧}{\ensuremath\rrbracket}
\newunicodechar{◅}{\ensuremath\triangleleft}

\input{macros}
\input{is-macros}

\begin{document}

\title[Inference Systems with Corules for Safety and Liveness Properties]{Inference Systems with Corules for Combined Safety\texorpdfstring{\\}{} and Liveness Properties of Binary Session Types}

\titlecomment{{\lsuper*} Extended version of a paper appeared in the proceedings of ICALP 2021.}

\author[L.~Ciccone]{Luca Ciccone}
\author[L.~Padovani]{Luca Padovani}

\address{Universit\`a di Torino}
\email{luca.ciccone@unito.it, luca.padovani@unito.it}

\input{abstract}

\maketitle

\input{introduction}

\input{is}

\input{types}

\input{termination}

\input{compliance}

\input{subtyping}

\input{agda}

\input{conclusion}

\section*{Acknowledgement}

We are grateful to Francesco Dagnino for his feedback on an early
version of this paper and to the anonymous reviewers for their
comments and suggestions that helped us improving the paper.

\bibliographystyle{alphaurl}
\bibliography{main}

\end{document}

%% file: macros.tex
\newif\ifproofs
\proofstrue

\newif\ifcomments
\commentstrue

\definecolor{lucared}{rgb}{0.5,0,0}
\definecolor{lucagreen}{rgb}{0,0.75,0}
\definecolor{lucablue}{rgb}{0,0,0.5}

\newcommand{\rulename}[1]{\textsc{\small[#1]}}
\newcommand{\defrule}[1]{\hypertarget{rule:#1}{\rulename{#1}}}
\newcommand{\refrule}[1]{\hyperlink{rule:#1}{\rulename{#1}}}

\newcommand{\eoe}{\hfill$\blacksquare$}
\newcommand{\eor}{\eoe}

\renewcommand{\DefTirName}[1]{\defrule{#1}}

\newcommand{\infercorule}[3][]{%
  {\mprset{fraction={===}}\inferrule[#1]{#2}{#3}}%
}

\let\oldmarginpar\marginpar
\renewcommand{\marginpar}[1]{\oldmarginpar[\raggedleft\tiny #1]%
{\raggedright\tiny #1}}

\newcommand{\action}{\actionA}
\newcommand{\actionA}{\alpha}
\newcommand{\actionB}{\beta}
\newcommand{\actionC}{\gamma}

\newcommand{\actions}{\actionsA}
\newcommand{\actionsA}{\varphi}
\newcommand{\actionsB}{\psi}

\newcommand{\es}{\varepsilon}

\newcommand{\Message}{\mathbb{V}}

\newcommand{\Nat}{\mathbb{N}}
\newcommand{\NatPlus}{\Nat^{\texttt{+}}}

\newcommand{\vtrue}{\mathsf{true}}
\newcommand{\vfalse}{\mathsf{false}}
\newcommand{\Bool}{\mathbb{B}}

\newcommand{\x}{x}

\newcommand{\SessionType}{\mathbb{S}}

\newcommand{\FairlyTerminating}{\mathbb{F}}

\newcommand{\T}{T}
\renewcommand{\S}{S}
\newcommand{\R}{R}

\newcommand{\End}[1][\In]{#1\mathsf{\color{lucablue}end}}
\newcommand{\Win}{\End[\Out]}
\newcommand{\TNil}{\mathsf{\color{lucared}nil}}
\newcommand{\Out}{{!}}
\newcommand{\In}{{?}}

\newcommand{\branch}{+}

\newcommand{\session}[2]{#1 \mathrel{\#} #2}

\newcommand{\co}[1]{\overline{#1}}

\newcommand{\traces}[1]{\mathtt{tr}(#1)}

\newcommand{\set}[1]{\{#1\}}

\newcommand{\annotation}[1]{\mathsf{#1}}
\newcommand{\Ind}{\annotation{Ind}}

\newcommand{\terminates}[2][]{#2{\Downarrow_{#1}}}

\newcommand{\compliance}[3][]{#2 \dashv_{#1} #3}

\newcommand{\subt}[3][]{#2 \leqslant_{#1} #3}
\newcommand{\isubt}{\subt[\Ind]}

\newcommand{\eqdef}{\stackrel{\text{\tiny\sf def}}=}

\newcommand{\red}{\to}
\newcommand{\nred}{\arrownot\red}
\newcommand{\wred}{\Rightarrow}

\newcommand{\lred}[1]{\stackrel{#1}\longrightarrow}

\newcommand{\xlred}[1]{\xrightarrow{#1}}

\newcommand{\wlred}[1]{\stackrel{#1}\Longrightarrow}



%% file: is-macros.tex
\newcommand{\universe}{\mathcal{U}}
\newcommand{\is}[1][I]{\mathcal{#1}}
\newcommand{\cois}[1][I]{\is[#1]_{\mathsf{co}}}
\newcommand{\judg}{\mathit{j}} 
\newcommand{\InfOp}[1]{\mathit{F}_{#1}} 
\newcommand{\gis}[2]{\langle#1,#2\rangle}

\newcommand{\RulePair}[2]{\Pair{#1}{#2}} 
\newcommand{\Rule}[2]{\inferrule{#1}{#2}}

\newcommand{\prem}{\mathit{pr}}

\newcommand{\sem}[1]{\llbracket #1 \rrbracket} 
\newcommand{\Inductive}[1]{\textup{\textsf{Ind}}\sem{#1}}
\newcommand{\CoInductive}[1]{\textup{\textsf{CoInd}}\sem{#1}}
\newcommand{\FlexCo}[2]{\textup{\textsf{Gen}}\sem{#1,#2}}

\newcommand{\Trees}[1]{\mathcal{T}}

\newcommand{\Spec}{\mathcal{S}}

\newcommand{\mis}{\mathcal{I}} 
\newcommand{\mcois}{\mis_{\mathsf{co}}}

%% file: abstract.tex
\begin{abstract}
  Many properties of communication protocols combine \emph{safety} and
  \emph{liveness} aspects. Characterizing such combined properties by means of a
  single inference system is difficult because of the fundamentally different
  techniques (coinduction and induction, respectively) usually involved in
  defining and proving them.
  In this paper we show that \emph{Generalized Inference Systems} allow us to
  obtain sound and complete characterizations of (at least some of) these
  combined inductive/coinductive properties of binary session types.  In
  particular, we illustrate the role of \emph{corules} in characterizing fair
  termination (the property of protocols that can always eventually terminate),
  fair compliance (the property of interactions that can always be extended to
  reach client satisfaction) and fair subtyping, a liveness-preserving
  refinement relation for session types.
  The characterizations we obtain are simpler compared to the previously
  available ones and corules provide insight on the liveness properties being
  ensured or preserved. Moreover, we can conveniently appeal to the
  \emph{bounded coinduction principle} to prove the completeness of the provided
  characterizations.
\end{abstract}


%% file: introduction.tex
\newcommand{\nil}{[]}
\newcommand{\cons}[2]{#1 \mathbin{::} #2}
\newcommand{\maxlist}{\mathit{maxElem}}
\newcommand{\GIS}{GIS\xspace}
\newcommand{\GISs}{GISs\xspace}

\section{Introduction}
\label{sec:introduction}

Analysis techniques for concurrent and distributed programs usually
make a distinction between \emph{safety properties} --- ``nothing bad ever
happens'' --- and \emph{liveness properties} --- ``something good eventually
happens''~\cite{OwickiLamport82}.
For example, in a network of communicating processes, the absence of
communication errors and of deadlocks are safety properties, whereas the fact
that a protocol or a process can always successfully terminate is a liveness
property.
Because of their different nature, characterizations and proofs of safety and
liveness properties rely on fundamentally different (dual) techniques: safety
properties are usually based on invariance (coinductive) arguments, whereas
liveness properties are usually based on well foundedness (inductive)
arguments~\cite{AlpernSchneider85,AlpernSchneider87}.

The correspondence and duality between safety/coinduction and liveness/induction
is particularly apparent when properties are specified as formulas in the modal
\emph{$\mu$-calculus}~\cite{Kozen83,Stirling01,BradfieldStirling07}, a modal
logic equipped with least and greatest fixed points: safety properties are
expressed in terms of greatest fixed points, so that the ``bad things'' are
ruled out along all (possibly infinite) program executions; liveness properties
are expressed in terms of least fixed points, so that the ``good things'' are
always within reach along all program executions. Since the $\mu$-calculus
allows least and greatest fixed points to be interleaved arbitrarily, it makes
it possible to express properties that combine safety and liveness aspects,
although the resulting formulas are sometimes difficult to understand.

A different way of specifying (and enforcing) properties is by means of
\emph{inference systems}~\cite{Aczel77}. Inference systems admit two natural
interpretations, an inductive and a coinductive one respectively corresponding
to the least and the greatest fixed points of their associated inference
operator. Unlike the $\mu$-calculus, however, they lack the flexibility of
mixing their different interpretations since the inference rules are interpreted
either all inductively or all coinductively. For this reason, it is generally
difficult to specify properties that combine safety and liveness aspects by
means of a single inference system.
\emph{Generalized Inference Systems} (\GISs)~\cite{AnconaDagninoZucca17,Dagnino19} admit a wider range of interpretations,
including intermediate fixed points of the inference operator associated with
the inference system different from the least or the greatest one.  This is made
possible by the presence of \emph{corules}, whose purpose is to provide an
inductive definition of a space within which a coinductive definition is used.
Although \GISs do not achieve the same flexibility of the modal $\mu$-calculus
in combining different fixed points, they allow for the specification of
properties that can be expressed as the intersection of a least and a greatest
fixed point. This feature of \GISs resonates well with one of the fundamental
results in model checking stating that every property can be decomposed into a
conjunction of a safety property and a liveness
one~\cite{AlpernSchneider85,AlpernSchneider87,BaierKatoen08}.
The main contribution of this work is the realization that we can leverage this
decomposition result to provide compact and insightful characterizations of a
number of combined safety and liveness properties of binary session types using
\GISs.

Session types~\cite{Honda93,HondaVasconcelosKubo98,AnconaEtAl16,HuttelEtAl16}
are type-level specifications of communication protocols describing the allowed
sequences of input/output operations that can be performed on a communication
channel. By making sure that programs adhere to the session types of the
channels they use, and by establishing that these types enjoy particular
properties, it is possible to conceive compositional forms of analysis and
verification for concurrent and distributed programs.
In this work we illustrate the effectiveness of \GISs in characterizing the
following session type properties and relations:
\begin{description}
  \item[Fair termination] the property of protocols that can always eventually
terminate;
  \item[Fair compliance] the property of client/server interactions that can
    satisfy the client;
  \item[Fair subtyping] a liveness-preserving refinement relation for session
types.
\end{description}

We show how to provide sound and complete characterizations of these properties
just by adding a few corules to the inference systems of their ``unfair''
counterparts, those focusing on safety but neglecting liveness. Not only corules
shed light on the liveness(-preserving) property of interest, but we can
conveniently appeal to the \emph{bounded coinduction principle} of
\GISs~\cite{AnconaDagninoZucca17} to prove the completeness of the provided
characterizations, thus factoring out a significant amount of work.
We also make two side contributions. First, we provide an Agda~\cite{Norell07}
formalization of all the notions and results stated in the paper. In particular,
we give the first machine-checked formalization of a liveness-preserving
refinement relation for session types.
Second, the Agda representation of session types we adopt allows us to address a
family of \emph{dependent session
types}~\cite{ToninhoCairesPfenning11,ToninhoYoshida18,ThiemannVasconcelos20,CicconePadovani20}
in which the length and structure of the protocol may depend in non-trivial ways
on the \emph{content} of exchanged messages. Thus, we extend previously given
characterizations of fair compliance and fair
subtyping~\cite{Padovani13,Padovani16} to a much larger class of protocols.

\paragraph{Structure of the paper.}
We quickly recall the key definitions of \GISs in Section~\ref{sec:is} and
describe syntax and semantics of session types in Section~\ref{sec:types}. We
then define and characterize fair termination (Section~\ref{sec:termination}),
fair compliance (Section~\ref{sec:compliance}) 
and fair subtyping (Section~\ref{sec:subtyping}).
Section~\ref{sec:agda} provides a walkthrough of the Agda formalization of fair
compliance and its specification as a \GIS.\@ The full Agda formalization is
accessible from a public repository~\cite{CicconePadovani21}.
We conclude in Section~\ref{sec:conclusion}.

\paragraph{Origin of the paper.}
This is a revised and extended version of a paper that appears in the
proceedings of the \emph{48\textsuperscript{th} International Colloquium on
Automata, Languages, and Programming (ICALP'21)}~\cite{CicconePadovani21B}. The
most significant differences between this and the previous version of the paper
are summarized below:
\begin{itemize}
  \item The representation of session types in this version of the paper differs
  from that in the Agda formalization and is more aligned with traditional ones.
  In the previous version, the representation of session types followed closely
  the Agda formalization, and that was a potential source of confusion
  especially for readers not acquainted with Agda.
  \item Example derivations for the presented \GISs have been added.
  \item ``Pen-and-paper'' proof sketches of the presented result have been
  added/expanded.
  \item The detailed walkthrough of the Agda formalization of fair compliance
  (Section~\ref{sec:agda}) is new.
\end{itemize}
Finally, the Agda formalization has been upgraded to a more recent version of
the \GIS library for Agda~\cite{CicconeDagninoZucca21} that supports (co)rules
with infinitely many premises. This feature allows us to provide a formalization
that is fully parametric on the set of values that can be exchanged over a
session. In contrast, the formalization referred to from the previous version of
the paper was limited to boolean values.


%% file: is.tex
\newcommand{\Pair}[2]{\langle#1,#2\rangle}
\newcommand{\fun}[2]{#1 \to #2}
\newcommand{\ple}[1]{\langle#1\rangle}

\section{Generalized Inference Systems}
\label{sec:is}

In this section we briefly recall the key notions of Generalized Inference
Systems (\GISs). In particular, we see how \GISs enable the definition of
predicates whose purely (co)inductive interpretation does not yield the intended
meaning and we review the canonical technique to prove the completeness of a
defined predicate with respect to a given specification. Further details on
\GISs may be found in the existing
literature~\cite{AnconaDagninoZucca17,Dagnino19}.

An \emph{inference system}~\cite{Aczel77} $\mis$ over a \emph{universe}
$\universe$ of \emph{judgments} is a set of \textit{rules}, which are pairs
$\RulePair\prem\judg$ where $\prem\subseteq\universe$ is the set of
\emph{premises} of the rule and $\judg\in\universe$ is the \emph{conclusion} of
the rule. A rule without premises is called \emph{axiom}. Rules are typically
presented using the syntax
\[
  \Rule\prem\judg
\]
where the line separates the premises (above the line) from the
conclusion (below the line).

\begin{rem}
  In many cases, and in this paper too, it is convenient to present inference
  systems using \emph{meta-rules} instead of rules. A meta-rule stands for a
  possibly infinite set of rules which are obtained by instantiating the
  \emph{meta-variables} occurring in the meta-rule.
  In the rest of the paper we will not insist on this distinction and we will
  use ``(co)rule'' even when referring to meta-(co)rules. If necessary, we will
  use side conditions to constrain the valid instantiations of the
  meta-variables occurring in such meta-(co)rules.
  \eor
\end{rem}

A \emph{predicate} on $\universe$ is any subset of $\universe$. An
\emph{interpretation} of an inference system $\mis$ identifies a
\emph{predicate} on $\universe$ whose elements are called \emph{derivable
judgments}. To define the interpretation of an inference system $\mis$, consider
the \emph{inference operator} associated with $\mis$, which is the function
$\InfOp\mis : \wp(\universe) \to \wp(\universe)$ such that
\[
  \InfOp\mis(X)=\set{\judg\in\universe \mid \exists\prem\subseteq X: \RulePair\prem\judg\in\mis}
\]
for every $X \subseteq \universe$.
Intuitively, $\InfOp\mis(X)$ is the set of judgments that can be
derived in one step from those in $X$ by applying a rule of $\mis$.
Note that $\InfOp\mis$ is a monotone endofunction on the complete
lattice $\wp(\universe)$, hence it has least and greatest fixed
points.

\begin{defi}
  The \emph{inductive interpretation} $\Inductive\mis$ of an
  inference system $\mis$ is the least fixed point of $\InfOp\mis$
  and the \emph{coinductive interpretation} $\CoInductive\mis$ is
  the greatest one.
\end{defi}

From a proof theoretical point of view, $\Inductive\mis$ and $\CoInductive\mis$
are the sets of judgments derivable with well-founded and non-well-founded proof
trees, respectively.

Generalized Inference Systems enable the definition of (some)
predicates for which neither the inductive interpretation nor the
coinductive one give the expected meaning.

\begin{defi}[generalized inference system]
  \label{def:gis}
  A \emph{generalized inference system} is a pair $\Pair\mis\mcois$ where $\mis$
  and $\mcois$ are inference systems (over the same $\universe$) whose elements
  are called \emph{rules} and \emph{corules}, respectively.
  The interpretation of a generalized inference system $\Pair\mis\mcois$,
  denoted by $\FlexCo\mis\mcois$, is the greatest post-fixed point of
  $\InfOp\mis$ that is included in $\Inductive{\mis\cup\mcois}$.
\end{defi}

From a proof theoretical point of view, a \GIS $\ple{\mis,\mcois}$ identifies
those judgments derivable with an arbitrary (not necessarily well-founded) proof
tree in $\mis$ and whose nodes (the judgments occurring in the proof tree) are
derivable with a well-founded proof tree in $\mis\cup\mcois$.

Consider now a \emph{specification} $\Spec\subseteq\universe$, that is an
arbitrary subset of $\universe$. We can relate $\Spec$ to the interpretation of
a (generalized) inference system using one of the following proof principles.
The \emph{induction principle}~\cite[Corollary 2.4.3]{Sangiorgi11} allows us to
prove the \emph{soundness} of an inductively defined predicate by showing that
$\Spec$ is \emph{closed} with respect to $\mis$. That is, whenever the premises
of a rule of $\mis$ are all in $\Spec$, then the conclusion of the rule is also
in $\Spec$.

\begin{prop}
  \label{prop:ip}
  If $\InfOp\mis(\Spec) \subseteq \Spec$, then $\Inductive\mis \subseteq \Spec$.
\end{prop}

The \emph{coinduction principle}~\cite[Corollary 2.4.3]{Sangiorgi11} allows us
to prove the \emph{completeness} of a coinductively defined predicate by showing
that $\Spec$ is \emph{consistent} with respect to $\mis$.  That is, every
judgment of $\Spec$ is the conclusion of a rule whose premises are also in
$\Spec$.

\begin{prop}
  \label{prop:cp}
  If $\Spec \subseteq \InfOp\mis(\Spec)$, then $\Spec \subseteq
  \CoInductive\mis$.
\end{prop}

The \emph{bounded coinduction principle}~\cite{AnconaDagninoZucca17} allows us
to prove the \emph{completeness} of a predicate defined by a generalized
inference system $\ple{\mis,\mcois}$. In this case, one needs to show not only
that $\Spec$ is consistent with respect to $\mis$, but also that $\Spec$ is
\emph{bounded} by the inductive interpretation of the inference system $\mis
\cup \mcois$. Formally:

\begin{prop}
  \label{prop:bcp}
  If $\Spec \subseteq \Inductive{\mis\cup\mcois}$ and
  $\Spec \subseteq \InfOp\mis(\Spec)$, then
  $\Spec\subseteq\FlexCo\mis\mcois$.
\end{prop}

Proving the boundedness of $\Spec$ amounts to proving the completeness of
$\mis\cup\mcois$ (inductively interpreted) with respect to $\Spec$. All of the
\GISs that we are going to discuss in
Sections~\ref{sec:termination}--\ref{sec:subtyping} are proven complete using
the bounded coinduction principle.

\begin{exa}[maximum of a colist]
  \label{ex:maxelem}
  A recurring example in the literature of \GISs is the predicate
  $\maxlist(l,x)$, asserting that $x$ is the maximum element of a colist
  (a possibly infinite list) $l$. If we consider the inference system
  \begin{equation}
    \label{eq:maxlist}
    \inferrule{ }{
      \maxlist(\cons\x\nil, \x)
    }
    \qquad
    \inferrule{
      \maxlist(l, y)
    }{
      \maxlist(\cons\x{l}, \max\set{x,y})
    }
  \end{equation}
  where $\nil$ denotes the empty list and $\cons{}{}$ is the constructor, we
  observe that neither of its two natural interpretation gives the intended
  meaning to $\maxlist$.
  Indeed, the \emph{inductive} interpretation of the rules restricts the set of
  derivable judgments to those for which there is a \emph{well-founded}
  derivation tree. In this case, the $\maxlist$ predicate is sound but not
  complete, since it does not hold for any infinite colist, even those for which
  the maximum exists. The \emph{coinductive} interpretation of these rules
  allows us to derive judgments by means of \emph{non-well-founded} derivation
  trees. In this case, the $\maxlist$ predicate is complete but not sound. In
  particular, it becomes possible to derive any judgment $\maxlist(l, x)$ where
  $x$ is greater than the elements of the colist, but is not an element of the
  colist. For example, if $l = \cons1{l}$ is an infinite colist of $1$'s, the
  infinite derivation
  \[
    \begin{prooftree}
      \[
        \smash\vdots\mathstrut
        \justifies
        \maxlist(l,2)
      \]
      \justifies
      \maxlist(l,2)
    \end{prooftree}
  \]
  allows us to conclude that $2$ is the maximum of $l$, even though $2$ does not occur in $l$.
  
  To repair the above inference system we can add the following \emph{coaxiom}:
  \begin{equation}
    \label{eq:comaxlist}
    \infercorule{ }{
      \maxlist(\cons\x{l}, x)
    }
  \end{equation}
  
  Read naively, this coaxiom seems to assert that the first element of any
  colist is also its maximum. In the context of a \GIS, its actual effect is
  that of ruling out those judgments $\maxlist(l,x)$ in which $x$ is \emph{not}
  an element of the colist.
  Indeed, the inductive interpretation of the rules~\eqref{eq:maxlist} and the
  coaxiom~\eqref{eq:comaxlist} is the space of judgments $\maxlist(l,x)$ such
  that $x$ is an element of $l$. Then, the generalized interpretation of the
  \GIS is defined as the coinductive interpretation of the
  rules~\eqref{eq:maxlist} within this space.
  In other words, the coaxiom~\eqref{eq:comaxlist} adds a \emph{well-foundedness
  element} to the derivability of a judgment $\maxlist(l,x)$, by requiring that
  $x$ must be found in --- at some finite distance from the head of --- the colist
  $l$.
  \eoe
\end{exa}

\begin{rem}
  The terminology we adopt for \GISs may be misleading because the word
  \emph{co}rule seems to suggest that the rule is interpreted
  \emph{co}inductively, whereas corules play a role in the inductive
  interpretation of \GISs (Definition~\ref{def:gis}). The confusion is
  reinforced by the choice of notation, whereby corules are distinguished by a
  double line. This notation has been sometimes used in the existing literature
  for denoting coinductively interpreted inference rules. In this paper we have
  chosen to stick with the terminology and the notation used in the works that
  have introduced \GISs~\cite{AnconaDagninoZucca17,Dagnino19}.
  \eor
\end{rem}


%% file: types.tex
\section{Syntax and Semantics of Session Types}
\label{sec:types}

We assume a set $\Message$ of \emph{values} that can be exchanged in
communications. This set may include booleans, natural numbers, strings, and so
forth. Hereafter, we assume that $\Message$ contains at least \emph{two}
elements, otherwise branching protocols cannot be described and the theoretical
development that follows becomes trivial. We use $x$, $y$, $z$ to range over the
elements of $\Message$.

We define the set $\SessionType$ of \emph{session types} over $\Message$ using
coinduction, to account for the possibility that session types (and the
protocols they describe) may be infinite.

\begin{defi}[session types]
  \label{def:st}
  Session types $T$, $S$ are the possibly infinite trees coinductively generated
  by the productions
  \[
    \begin{array}{rrcl}
      \textbf{Polarity}       & p, q & \in & \set{{\In}, {\Out}} \\
      \textbf{Session type} & T, S & ::= & \TNil \mid p\set{x:T_x}_{x\in\Message} \\
    \end{array}
  \]
\end{defi}

A session type describes the valid sequences of input/output actions that can be
performed on a communication channel. We use polarities to discriminate between
input actions ($\In$) and output actions ($\Out$). Hereafter, we write $\co{p}$
for the opposite or dual polarity of $p$, that is $\co\In = \Out$ and $\co\Out =
\In$.
An \emph{input session type} $\In\set{x:T_x}_{x\in\Message}$ describes a channel
used first for receiving a message $x\in\Message$ and then according to the
continuation $T_x$.
Dually, an \emph{output session type} $\Out\set{x:T_x}_{x\in\Message}$ describes
a channel used first for sending a message $x\in\Message$ and then according to
$T_x$.
Note that input and output session types specify continuations for \emph{all}
possible values in the set $\Message$. The session type $\TNil$, which describes
an \emph{unusable} session channel, can be used as continuation for those values
that \emph{cannot} be received or sent. As we will see shortly, the presence of
$\TNil$ breaks the symmetry between inputs and outputs.

It is convenient to introduce some notation for presenting session types in a
more readable and familiar form.
Given a polarity $p$, a set $X \subseteq \Message$ of values and a family
$T_{x\in X}$ of session types, we let
\[
  p\set{x:T_x}_{x\in X} \eqdef p\left(\set{x:T_x}_{x\in X} \cup \set{x:\TNil}_{x\in\Message\setminus X}\right)
\]
so that we can omit explicit $\TNil$ continuations. As a special case when all
the continuations are $\TNil$, we write $\End[p\,]$ instead of $p\emptyset$.
Both $\End$ and $\Win$ describe session channels on which no further
communications may occur, although they differ slightly with respect to the
session types they can be safely combined with.  Describing terminated protocols
as degenerate cases of input/output session types reduces the amount of
constructors needed for their Agda representation (Section~\ref{sec:agda}).
Another common case for which we introduce a convenient notation is when the
continuations are the same, regardless of the value being exchanged: in these
cases, we write $pX.T$ instead of $p\set{x:T}_{x\in X}$. For example,
$\Out\Bool.\T$ describes a channel used for sending a boolean and then according
to $T$ and $\In\Nat.\S$ describes a channel used for receiving a natural number
and then according to $S$. We abbreviate $\set\x$ with $\x$ when no confusion
may arise. So we write $\Out\vtrue.\T$ instead of $\Out\set\vtrue.\T$.

Finally, we define a partial operation $+$ on session types such that
\[
  p\set{x:T_x}_{x\in X} + p\set{x:T_x}_{x\in Y}
  \eqdef
  p\set{x:T_x}_{x\in X\cup Y}
\]
when $X \cap Y = \emptyset$. For example, $\Out\vtrue.S_1 + \Out\vfalse.S_2$
describes a channel used first for sending a boolean value and then according to
$S_1$ or $S_2$ depending on the boolean value.
It it easy to see that $+$ is commutative and associative and that $\End[p\,]$
is the unit of $+$ when used for combining session types with polarity $p$. Note
that $T + S$ is undefined if the topmost polarities of $T$ and $S$ differ.
We assume that $+$ binds less tightly than the `$.$' in continuations.

We do not introduce any concrete syntax for specifying infinite session types.
Rather, we specify possibly infinite session types as solutions of equations of
the form $S = \cdots$ where the metavariable $S$ may also occur (guarded) on the
right-hand side of `$=$'. Guardedness guarantees that the session type $S$
satisfying such equation does exist and is unique~\cite{Courcelle83}.

\begin{exa}
  \label{ex:running}
  The session types $T_1$ and $S_1$ that satisfy the equations
  \[
    T_1 = \Out\vtrue.\Out\Nat.\T_1 + \Out\vfalse.\End
    \qquad
    S_1 = \Out\vtrue.\Out\NatPlus.\S_1 + \Out\vfalse.\End
  \]
  both describe a channel used for sending a boolean. If the boolean
  is $\vfalse$, the communication stops immediately ($\End$). If it
  is $\vtrue$, the channel is used for sending a natural number (a
  strictly positive one in $S_1$) and then according to $T_1$ or
  $S_1$ again. Notice how the structure of the protocol after the
  output of the boolean depends on the \emph{value} of the boolean.

  The session types $T_2$ and $S_2$ that satisfy the equations
  \[
    T_2 = \In\vtrue.\Out\Nat.\T_2 \branch \In\vfalse.\End
    \qquad
    S_2 = \In\vtrue.\Out\NatPlus.\S_2 \branch \In\vfalse.\End
  \]
  differ from $T_1$ and $S_1$ in that the channel they describe is
  used initially for \emph{receiving} a boolean.
  \eoe
\end{exa}

We define the operational semantics of session types by means of a
\emph{labeled transition system}. \emph{Labels}, ranged over by
$\actionA$, $\actionB$, $\actionC$, have either the form $\In\x$
(input of message $x$) or the form $\Out\x$ (output of message $x$).
Transitions $T \lred\action S$ are defined by the following axioms:
\begin{equation}
  \label{eq:lts}
  \inferrule[input]{\mathstrut}{
    \In\x.S + T \lred{\In\x} S
  }
  \qquad
  \inferrule[output]{\mathstrut}{
    \Out\x.S + T \lred{\Out\x} S
  }
  ~
  S \ne \TNil
\end{equation}

There is a fundamental asymmetry between send and receive operations: the act of
sending a message is \emph{active} --- the sender may choose the message to send
--- while the act of receiving a message is \emph{passive} --- the receiver cannot
cherry-pick the message being received.
We model this asymmetry with the side condition $S \ne \TNil$ in
\refrule{output} and the lack thereof in~\refrule{input}:
a process that uses a session channel according to
$\Out\set{x:T_x}_{x\in\Message}$ refrains from sending a message $x$ if $T_x =
\TNil$, namely if the channel becomes unusable by sending that particular
message, whereas a process that uses a session channel according to
$\In\set{x:T_x}_{x\in\Message}$ cannot decide which message $x$ it will receive,
but the session channel becomes unusable if an unexpected message arrives. The
technical reason for modeling this asymmetry is that it allows us to capture a
realistic communication semantics and will be discussed in more detailed in
Remark~\ref{rem:asymmetry}. For the time being, these transition rules allow us
to appreciate a little more the difference between $\Win$ and $\End$. While both
describe a session endpoint on which no further communications may occur, $\Win$
is ``more robust'' than $\End$ since it has no transitions, whereas $\End$ is
``more fragile'' than $\Win$ since it performs transitions, all of which lead to
$\TNil$. For this reason, we use $\Win$ to flag successful session termination
(Section~\ref{sec:compliance}), whereas $\End$ only means that the protocol has
ended.

To describe \emph{sequences} of consecutive transitions performed by
a session type we use another relation $\wlred\actions$ where
$\actions$ and $\actionsB$ range over strings of labels. As usual,
$\es$ denotes the empty string and juxtaposition denotes string
concatenation. The relation $\wlred\actions$ is the least one such
that $T \wlred\es T$ and if $T \lred\action S$ and
$S \wlred\actions R$, then $T \wlred{\action\actions} R$.


%% file: termination.tex
\section{Fair Termination}
\label{sec:termination}

We say that a session type is fairly terminating if it preserves the possibility
of reaching $\Win$ or $\End$ along all of its transitions that do not lead to
$\TNil$.  Fair termination of $T$ does not necessarily imply that there exists an
upper bound to the length of communications that follow the protocol $T$, but it
guarantees the absence of ``infinite loops'' whereby the communication is forced
to continue forever.

To formalize fair termination we need the notion of \emph{trace}, which is a
finite sequence of actions performed on a session channel while preserving
usability of the channel.

\begin{defi}[traces and maximal traces]
  \label{def:traces}
  The \emph{traces} of $T$ are defined as $\traces\T \eqdef \set{ \actions \mid
  \exists\S : T \wlred\actions S \ne \TNil }$.  We say that $\actions \in
  \traces\T$ is \emph{maximal} if $\actions\actionsB \in \traces\T$ implies
  $\actionsB = \es$.
\end{defi}

For example, we have $\traces\TNil = \emptyset$ and $\traces\Win = \traces\End =
\set\es$. Note that $\Win$ and $\End$ have the same traces but different
transitions (hence different behaviors).
A \emph{maximal trace} is a trace that cannot be extended any further. For
example $\es$ is a maximal trace of both $\Win$ and $\End$ but not of
$\Out\Bool.\End$ whereas $\Out\vtrue$ and $\Out\vfalse$ are maximal traces of
$\Out\Bool.\End$.

\begin{defi}[fair termination]
  \label{def:wt}
  We say that $T$ is \emph{fairly terminating} if, for every
  $\actionsA\in\traces\T$, there exists $\actionsB$ such that
  $\actionsA\actionsB \in \traces\T$ and $\actionsA\actionsB$ is maximal.
\end{defi}

\begin{exa}
  \label{ex:termination}
  All of the session types presented in Example~\ref{ex:running} are fairly
  terminating. The session type $R = \Out\Bool.R$, which describes a channel
  used for sending an infinite stream of boolean values, is not fairly
  terminating because no trace of $R$ can be extended to a maximal one. Note
  that also $R' \eqdef \Out\vtrue.R + \Out\vfalse.\Win$ is not fairly
  terminating, even though there is a path leading to $\Win$, because fair
  termination must be \emph{preserved} along all possible transitions of the
  session type, whereas $R' \xlred{\Out\vtrue} R$ and $R$ is not fairly
  terminating.
  Finally, $\TNil$ is trivially fairly terminating because it has no trace.
  \eoe
\end{exa}

\begin{table}
  \[
    \begin{array}{@{}c@{}}
      \inferrule[t-nil]{\mathstrut}{
        \terminates\TNil
      }
      \qquad
      \inferrule[t-all]{
        \terminates{T_x}~{}^{(\forall x\in\Message)}
      }{
        \terminates{p\set{x:T_x}_{x\in\Message}}
      }
      \qquad
      \infercorule[t-any]{
        \terminates{S}
      }{
        px.S + T
      }
      ~
      S \ne \TNil
    \end{array}
  \]
  \caption{
    \label{tab:ft}
    Generalized inference system $\gis{\is[T]}{\cois[T]}$ for fair termination.
    }
\end{table}

To find an inference system for fair termination observe that the set
$\FairlyTerminating$ of fairly terminating session types is the largest one that
satifies the following two properties:
\begin{enumerate}
  \item it must be possible to reach either $\Win$ or $\End$ from every $T \in
\FairlyTerminating \setminus \set\TNil$;
  \item the set $\FairlyTerminating$ must be closed by transitions, namely if $T
\in \FairlyTerminating$ and $T \lred\action S$ then $S \in \FairlyTerminating$.
\end{enumerate}
Neither of these two properties, taken in isolation, suffices to define
$\FairlyTerminating$: the session type $R'$ in Example~\ref{ex:termination}
enjoys property (1) but is not fairly terminating; the set $\SessionType$ is
obviously the largest one with property (2), but not every session type in it is
fairly terminating.
This suggests the definition of $\FairlyTerminating$ as the largest subset of
$\SessionType$ satisfying (2) and whose elements are \emph{bounded} by property
(1), which is precisely what corules allow us to specify.

Table~\ref{tab:ft} shows a \GIS $\gis{\is[T]}{\cois[T]}$ for fair termination,
where $\is[T]$ consists of all the (singly-lined) rules whereas $\cois[T]$
consists of all the (doubly-lined) corules (we will use this notation also in
the subsequent \GISs).
The axiom~\refrule{t-nil} indicates that $\TNil$ is fairly terminating in a
trivial way (it has no trace), while~\refrule{t-all} indicates that fair
termination is closed by all transitions. Note that these two rules, interpreted
coinductively, are satisfied by all session types, hence $\set{ \T \mid
\terminates\T \in \CoInductive{\is[T]} } = \SessionType$.

\begin{thm}
  \label{thm:termination}
  $T$ is fairly terminating if and only if $\terminates\T \in
  \FlexCo{\is[T]}{\cois[T]}$.
\end{thm}
\begin{proof}[Proof sketch]
  For the ``if'' part, suppose $\terminates\T \in \FlexCo{\is[T]}{\cois[T]}$ and
  consider a trace $\actions\in\traces\T$. That is, $T \wlred\actions S$ for
  some $S \ne \TNil$. Using~\refrule{t-all} we deduce $\terminates\S \in
  \FlexCo{\is[T]}{\cois[T]}$ by means of a simple induction on $\actions$. Now
  $\terminates\S \in \FlexCo{\is[T]}{\cois[T]}$ implies $\terminates\S \in
  \Inductive{\is[T] \cup \cois[T]}$ by Definition~\ref{def:gis}. Another
  induction on the (well-founded) derivation of this judgment, along with the
  witness message $x$ of~\refrule{t-any}, allows us to find $\actionsB$ such
  that $\actions\actionsB$ is a maximal trace of $T$.

  For the ``only if'' part, we apply the bounded coinduction principle (see
  Proposition~\ref{prop:bcp}). Since we have already argued that the coinductive
  interpretation of the GIS in Table~\ref{tab:ft} includes all session types, it
  suffices to show that $T$ fairly terminating implies $\terminates\T \in
  \Inductive{\is[T] \cup \cois[T]}$.
  From the assumption that $T$ is fairly terminating we deduce that there exists
  a maximal trace $\actions\in\traces{T}$. An induction on $\actions$ allows us
  to derive $\terminates{T}$ using repeated applications of~\refrule{t-any}, one
  for each action in $\actions$, topped by a single application of~\refrule{t-nil}.
\end{proof}

\begin{rem}
  \label{rem:termination}
  \newcommand{\len}[1]{\mathsf{len}(#1)}
  The notion of \emph{fair termination} we have given in Definition~\ref{def:wt}
  is easy to relate with the GIS in Table~\ref{tab:ft} but, in the existing
  literature, fair termination is usually formulated in a different way that
  justifies more naturally the fact that it is a termination property. The two
  formulations are equivalent, at least when we consider \emph{regular session
  types}, those whose tree is made of finitely many distinct subtrees.
  To elaborate, let a \emph{run} of $T$ be a sequence of transitions
  \[
    T = T_0 \lred{\action_1} T_1 \lred{\action_2} T_2 \lred{\action_3} \cdots
  \]
  such that no $T_i$ is $\TNil$ and it is said to be \emph{maximal} either if it
  is infinite or if the last session type in the sequence is either $\End$ or
  $\Win$.
  A run is \emph{fair}~\cite{Francez86,AptFrancezKatz87,GlabbeekHofner19} if,
  whenever some $S$ occurs infinitely often in it and $S \lred\action S' \ne
  \TNil$ is a transition from $S$ to $S'$, then this transition occurs
  infinitely often in the run. Intuitively, a fair run is one in which no
  transition that is enabled sufficiently (infinitely) often is discriminated
  against.

  It is not difficult to prove that $T$ is fairly terminating if and only if
  every maximal fair run of $T$ is finite. For the ``only if'' part, let
  \[
    \len\S \eqdef \min\set{ \actions \mid \text{$\actions$ is a maximal trace of $S$} }
  \]
  be the minimum length of any maximal trace of $S$, where we postulate that
  $\min\emptyset = \infty$.
  If $T$ is fairly terminating suppose, by contradiction, that $T$ has an
  infinite fair run $T = T_0 \lred{\action_1} T_1 \lred{\action_2} \cdots$. From
  the hypothesis that $T$ is fairly terminating we deduce that so is each $T_i$.
  Since $T$ is regular, there must be some subtree $S_0$ of $T$ that occurs
  infinitely often in the run. Such subtree cannot be $\End$ or $\Win$, since
  these session types have no successor in the run whereas $S_0$ has one. Among
  all the transitions of $S_0$, there must be one $S_0 \lred\action S_1$ such
  that $\len{S_1} < \len{S_0}$. From the hypothesis that the run is fair we
  deduce that $S_1$ occurs infinitely often in it. By repeating this argument we
  find an infinite sequence of session types $S_0, S_1, \dots$ such that
  $\len{S_{i+1}} < \len{S_i}$ for every $i$, which is absurd. We conclude that
  $T$ cannot have any infinite fair run.

  For the ``if'' part of the property we are proving, suppose that every maximal
  fair run of $T$ is finite and consider a trace $\actions\in\traces\T$. Then
  there exist $\action_1,\dots,\action_n$ and $T_1,\dots,T_n$ such that
  $\actions = \action_1\cdots\action_n$ and $T \lred{\action_1} T_1 \cdots
  \lred{\action_n} T_n$.
  It is a known fact that every finite run can be extended to a maximal fair run
  (this property is called \emph{machine closure}~\cite{Lamport00} or
  \emph{feasibility}~\cite{AptFrancezKatz87,GlabbeekHofner19}). Since we know
  that every maximal fair run of $T$ is finite, there exist $\actionB_1, \dots,
  \actionB_m$ and $S_1,\dots,S_m$ such that $T_n \lred{\actionB_1} S_1 \cdots
  \lred{\actionB_m} S_m \in \set{\End,\Win}$. Named $\actionsB$ the sequence
  $\actionB_1\cdots\actionB_m$, we conclude that $\actions\actionsB$ is a
  maximal trace of $T$.
  \eor
\end{rem}


%% file: compliance.tex
\section{Compliance}
\label{sec:compliance}

In this section we define and characterize two \emph{compliance} relations for
session types, which formalize the ``successful'' interaction between a client
and a server connected by a session. The notion of ``successful interaction''
that we consider is biased towards client satisfaction, but see
Remark~\ref{rem:duality} below for a discussion about alternative notions.

To formalize compliance we need to model the evolution of a session as client
and server interact. To this aim, we represent a session as a pair
$\session\R\T$ where $\R$ describes the behavior of the client and $T$ that of
the server.  Sessions reduce according to the rule
\begin{equation}
  \label{eq:sred}
  \inferrule{ }{
    \session\R\T \red \session{R'}{S'}
  }
  ~
  \text{if $R \lred{\co\action} R'$ and $T \lred\action T'$}
\end{equation}
where $\co\action$ is the \emph{complementary action} of $\action$ defined by
$\co{p\x} = \co{p}\x$.  We extend $\co{\,\cdot\,}$ to traces in the obvious way
and we write $\wred$ for the reflexive, transitive closure of $\red$. We write
$\session\R\T \red {}$ if $\session\R\T \red \session{R'}{T'}$ for some $R'$ and
$T'$ and $\session\R\T \nred$ if not $\session\R\T\red$.

The first compliance relation that we consider requires that, if the interaction
in a session stops, it is because the client ``is satisfied'' and the server
``has not failed'' (recall that a session type can turn into $\TNil$ only if an
unexpected message is received). Formally:

\begin{defi}[compliance]
  \label{def:comp}
  We say that $R$ is \emph{compliant} with $T$ if
  $\session\R\T \wred \session{\R'}{\T'} \nred$ implies $\R' = \Win$
  and $\T' \neq \TNil$.
\end{defi}

This notion of compliance is an instance of \emph{safety property} in which the
invariant being preserved at any stage of the interaction is that either client
and server are able to synchronize further, or the client is satisfied and the
server has not failed.

The second compliance relation that we consider adds a \emph{liveness}
requirement namely that, no matter how long client and server have been
interacting with each other, it is always possible to reach a configuration in
which the client is satisfied and the server has not failed.

\begin{defi}[fair compliance]
  \label{def:fcomp}
  We say that $R$ is \emph{fairly compliant} with $T$ if $\session\R\T \wred
  \session{\R'}{\T'}$ implies $\session{\R'}{\T'} \wred \session\Win{\T''}$ with
  $\T'' \neq \TNil$.
\end{defi}

It is easy to show that fair compliance implies compliance, but there exist
compliant session types that are not fairly compliant, as illustrated in the
following example.

\begin{exa}
  \label{ex:fcomp}
  Recall Example~\ref{ex:running} and consider the session types $R_1$ and $R_2$
  such that
  \[
    R_1 = \In\vtrue.\In\Nat.R_1 \branch \In\vfalse.\Win
    \qquad
    R_2 = \Out\vtrue.(\In0.\Win \branch \In\NatPlus.R_2)
  \]

  Then $R_1$ is fairly compliant with both $T_1$ and $S_1$ and $R_2$ is
  compliant with both $T_2$ and $S_2$.
  Even if $S_1$ exhibits fewer behaviors compared to $T_1$ (it never sends $0$
  to the client), at the beginning of a new iteration it can always send
  $\vfalse$ and steer the interaction along a path that leads $R_1$ to success.
  On the other hand, $R_2$ is fairly compliant with $T_2$ but not with $S_2$. In
  this case, the client insists on sending $\vtrue$ to the server in hope to
  receive $0$, but while this is possible with the server $T_2$, the server
  $S_2$ only sends strictly positive numbers.

  This example also shows that fair termination of both client and server is not
  sufficient, in general, to guarantee fair compliance. Indeed, both $R_2$ and
  $S_2$ are fairly terminating, but they are not fairly compliant. The reason is
  that the sequences of actions leading to $\Win$ on the client side are not
  necessarily the same (complemented) traces that lead to $\End$ on the server
  side. Fair compliance takes into account the synchronizations that can
  actually occur between client and server.
  \eoe
\end{exa}

\begin{rem}
  \label{rem:asymmetry}
  With the above notions of compliance we can now better motivate the
  asymmetric modeling of (passive) inputs and (active) outputs in the labeled
  transition system of session types~\eqref{eq:lts}. Consider the session
  types $R = \Out\vtrue.\Win + \Out\vfalse.\Out\vfalse.\Win$ and $S =
  \In\vtrue.\End$. Note that $R$ describes a client that succeeds by either
  sending a single $\vtrue$ value or by sending two $\vfalse$ values in
  sequence, whereas $S$ describes a server that can only receive a single
  $\vtrue$ value. If we add the same side condition $S \ne \TNil$ also for~\refrule{input} then $R$ would be compliant with $S$. Indeed, the server
  would be unable to perform the $\In\vfalse$-labeled transition, so that the
  only synchronization possible between $R$ and $S$ would be the one in which
  $\vtrue$ is exchanged. In a sense, with the $S \ne \TNil$ side condition in~\refrule{input} we would be modeling a communication semantics in which
  client and server \emph{negotiate} the message to be exchanged depending on
  their respective capabilities. Without the side condition, the message to be
  exchanged is always chosen by the active part (the sender) and, if the
  passive part (the receiver) is unable to handle it, the receiver fails. The
  chosen asymmetric communication semantics is also key to induce a notion of
  (fair) subtyping that is \emph{covariant} with respect to inputs
  (Section~\ref{sec:subtyping}).
  \eor
\end{rem}

\begin{table}
  \[
    \begin{array}{@{}c@{}}
      \inferrule[c-success]{\mathstrut}{
        \compliance\Win\T
      }
      ~T \neq \TNil
      \qquad
      \infercorule[c-sync]{
        \compliance{S}{T}
      }{
        \compliance{px.S + S'}{\co{p}x.T + T'}
      }
      \\\\
      \inferrule[c-inp-out]{
        \compliance{S_x}{T_x}~{}^{(\forall x\in X)}
      }{
        \compliance{\In\set{x:S_x}_{x\in\Message}}{\Out\set{x:T_x}_{x\in X}}
      }
      ~X \ne \emptyset
      \qquad
      \inferrule[c-out-inp]{
        \compliance{S_x}{T_x}~{}^{(\forall x\in X)}
      }{
        \compliance{\Out\set{x:S_x}_{x\in X}}{\In\set{x:T_x}_{x\in\Message}}
      }
      ~X \ne \emptyset
    \end{array}
  \]
  \caption{
    \label{tab:compliance}
    Generalized inference system $\gis{\is[C]}{\cois[C]}$ for fair compliance.}
\end{table}

Table~\ref{tab:compliance} presents the \GIS $\gis{\is[C]}{\cois[C]}$ for fair
compliance. Intuitively, a derivable judgment $\compliance{S}{T}$ means that the
client $S$ is (fairly) compliant with the server $T$.
Rule~\refrule{c-success} relates a satisfied client with a non-failed server.
Rules~\refrule{c-inp-out} and~\refrule{c-out-inp} require that, no matter which
message is exchanged between client and server, the respective continuations are
still fairly compliant. The side condition $X \ne \emptyset$ guarantees progress
by making sure that the sender is capable of sending at least one message.
As we will see, the coinductive interpretation of $\is[C]$, which consists of
these three rules, completely characterizes compliance
(Definition~\ref{def:comp}).
However, these rules do not guarantee that the interaction between client and
server can always reach a successful configuration as required by
Definition~\ref{def:fcomp}. For this, the corule~\refrule{c-sync} is essential.
Indeed, a judgment $\compliance{S}{T}$ that is derivable according to the
generalized interpretation of the \GIS $\gis{\is[C]}{\cois[C]}$ must admit a
well-founded derivation tree also in the inference system $\is[C] \cup
\cois[C]$. Since~\refrule{c-success} is the only axiom in this inference system,
finding a well-founded derivation tree in $\is[C] \cup \cois[C]$ boils down to
finding a (finite) path of synchronizations from $\session{S}{T}$ to a
successful configuration in which $S$ has reduced to $\Win$ and $T$ has reduced
to a session type other than $\TNil$.
Rule~\refrule{c-sync} allows us to find such a path by choosing the appropriate
messages exchanged between client and server.
In general, one can observe a dicotomy between the rules~\refrule{c-inp-out} and
\refrule{c-out-inp} having a universal flavor (they have many premises
corresponding to every possible interaction between client and server) and the
corule~\refrule{c-sync} having an existential flavor (it has one premise
corresponding to a particular interaction between client and server). This is
consistent with the fact that we use rules to express a safety property (which
is meant to be \emph{invariant} hence preserved by all the possible
interactions) and we use the corule to help us expressing a liveness property.
This pattern in the usage of rules and corules is quite common in \GISs because
of their interpretation and it can also be observed in the \GIS for fair
termination (Table~\ref{tab:ft}) and, to some extent, in that for fair subtyping
as well (Section~\ref{sec:subtyping}).

\begin{exa}
  \label{ex:fcomp:gis}
  Consider again $R_2 = \Out\vtrue.(\In0.\Win \branch \In\NatPlus.R_2)$ from
  Example~\ref{ex:fcomp} and $T_2 = \In\vtrue.\Out\Nat.\T_2 \branch
  \In\vfalse.\End$ from Example~\ref{ex:running}.
  In order to show that $\compliance{R_2}{T_2} \in \FlexCo{\is[C]}{\cois[C]}$ we
  have to find a possibly infinite derivation for $\compliance{R_2}{T_2}$ using
  the rules in $\is[C]$ as well as finite derivations for all of the judgments
  occurring in this derivation in $\is[C] \cup \cois[C]$.
  
  For the former we have
  \[
    \begin{prooftree}
      \[
        \[
          \justifies
          \compliance\Win{T_2}
          \using\refrule{c-success}
        \]
        \[
          \smash\vdots\mathstrut
          \justifies
          \compliance{R_2}{T_2}
        \]
        \justifies
        \compliance{
          \In0.\Win + \In\NatPlus.R_2
        }{
          \Out\Nat.T_2
        }
        \using\refrule{c-inp-out}
      \]
      \justifies
      \compliance{R_2}{T_2}
      \using\refrule{c-out-inp}
    \end{prooftree}
  \]
  where, in the application of~\refrule{c-inp-out}, we have collapsed all of the
  premises corresponding to the $\NatPlus$ messages into a single premise. Thus,
  we have proved $\compliance{R_2}{T_2} \in \CoInductive{\is[C]}$.
  Note the three judgments occurring in the above derivation tree. The finite
  derivation
  \[
    \begin{prooftree}
      \[
        \[
          \justifies
          \compliance\Win{T_2}
          \using\refrule{c-success}
        \]
        \justifies
        \compliance{
          \In0.\Win + \In\NatPlus.R_2
        }{
          \Out\Nat.T_2
        }
        \using\refrule{c-sync}
      \]
      \justifies
      \compliance{R_2}{T_2}
      \using\refrule{c-sync}
    \end{prooftree}
  \]
  shows that $\compliance{R_2}{T_2} \in \Inductive{\is[C] \cup \cois[C]}$. We
  conclude $\compliance{R_2}{T_2} \in \FlexCo{\is[C]}{\cois[C]}$.
  \eoe
\end{exa}

Observe that the corule~\refrule{c-sync} is at once essential and unsound. For
example, without it we would be able to derive the judgment
$\compliance{\R_2}{\S_2}$ despite the fact that $R_2$ is not fair compliant with
$S_2$ (Example~\ref{ex:fcomp}). At the same time, if we treated
\refrule{c-sync} as a plain rule, we would be able to derive the judgment
$\compliance{\Out\Nat.\Win}{\In0.\End}$ despite the reduction
$\session{\Out\Nat.\Win}{\In0.\End} \red \session\Win\TNil$ since \emph{there
exists} an interaction that leads to the successful configuration
$\session\Win\End$ (if the client sends $0$) but none of the others does.

\begin{thm}[compliance]
  \label{thm:compliance}
  For every $R, T \in \SessionType$, the following properties hold:
  \begin{enumerate}
  \item $R$ is compliant with $T$ if and only if $\compliance\R\T \in
    \CoInductive{\is[C]}$;
  \item $R$ is fairly compliant with $T$ if and only if
    $\compliance\R\T \in \FlexCo{\is[C]}{\cois[C]}$.
  \end{enumerate}
\end{thm}
\begin{proof}[Proof sketch]
  We sketch the proof of item~(2), which is the most interesting one. In
  Section~\ref{sec:agda} we describe the full proof formalized in Agda.
  For the ``if'' part, suppose that $\compliance\R\T \in
  \FlexCo{\is[C]}{\cois[C]}$ and consider a reduction $\session\R\T \wred
  \session{R'}{T'}$. An induction on the length of this reduction, along with~\refrule{c-inp-out} and~\refrule{c-out-inp}, allows us to deduce
  $\compliance{R'}{T'} \in \FlexCo{\is[C]}{\cois[C]}$. Then we have
  $\compliance{R'}{T'} \in \Inductive{\is[C] \cup \cois[C]}$ by
  Definition~\ref{def:gis}. An induction on this (well-founded) derivation
  allows us to find a reduction $\session{R'}{T'} \wred \session\Win{T''}$ such
  that $T'' \ne \TNil$.

  For the ``only if'' part we apply the bounded coinduction principle
  (Proposition~\ref{prop:bcp}). Concerning consistency, we show that whenever
  $R$ is fairly compliant with $T$ we have that $\compliance{R}{T}$ is the
  conclusion of a rule in Table~\ref{tab:compliance} whose premises are pairs of
  fairly compliant session types. Indeed, from the hypothesis that $R$ is fairly
  compliant with $T$ we deduce that there exists a derivation $\session{R}{T}
  \wred \session\Win{T'}$ for some $T'\ne\TNil$. A case analysis on the shape of
  $R$ and $T$ allows us to deduce that either $R = \Win$ and $T = T' \ne \TNil$,
  in which case the axiom~\refrule{c-success} applies, or that $R$ and $T$ must be
  input/output session types with opposite polarities such that the sender has
  at least one non-$\TNil$ continuation and whose reducts are still fairly
  compliant (because fair compliance is preserved by reductions). Then either~\refrule{c-inp-out} or~\refrule{c-out-inp} applies.
  Concerning boundedness, we do an induction on the reduction $\session{R}{T}
  \wred \session\Win{T'}$ to build a well-founded tree made of a suitable number
  of applications of~\refrule{c-sync} topped by a single application of~\refrule{c-success}.
\end{proof}


%% file: subtyping.tex
\section{Subtyping}
\label{sec:subtyping}

The notions of compliance given in Section~\ref{sec:compliance} induce
corresponding semantic notions of subtyping that can be used to define a safe
substitution principle for session types~\cite{LiskovWing94}. Intuitively, $\T$
is a subtype of $\S$ if any client that successfully interacts with $\T$ does so
with $\S$ as well. The key idea is the same used for defining testing
equivalences for processes~\cite{DeNicolaHennessy84,Hennessy88,RensinkVogler07},
except that we use the term ``client'' instead of the term ``test''. 

\begin{defi}[subtyping]
  \label{def:sub}
  We say that $T$ is a \emph{subtype} of $S$ if $R$ compliant with
  $T$ implies $R$ compliant with $S$ for every $R$.
\end{defi}

\begin{defi}[fair subtyping]
  \label{def:fsub}
  We say that $T$ is a \emph{fair subtype} of $S$ if $R$ fairly compliant with
  $T$ implies $R$ fairly compliant with $S$ for every $R$.
\end{defi}

According to these definitions, when $T$ is a (fair) subtype of $S$, a process
that behaves according to $T$ can be replaced by a process that behaves
according to $S$ without compromising (fair) compliance with the clients of $T$.
At first sight this substitution principle appears to be just the opposite of
the expected/intended one, whereby it is safe to use a session channel of type
$T$ where a session channel of type $S$ is expected if $T$ is a subtype of $S$.
The mismatch is only apparent, however, and can be explained by looking
carefully at the entities being replaced in the substitution principles recalled
above (processes in one case, session channels in the other). The interested
reader may refer to Gay~\cite{Gay16} for a study of these two different, yet
related viewpoints.
What matters here is that the above notions of subtyping are ``correct by
definition'' but do not provide any hint as to the shape of two session types
$T$ and $S$ that are related by (fair) subtyping. This problem is well known in
the semantic approaches for defining subtyping
relations~\cite{FrischCastagnaBenzaken08,BernardiHennessy16} as well as in the
aforementioned testing theories for
processes~\cite{DeNicolaHennessy84,Hennessy88,RensinkVogler07}, which the two
definitions above are directly inspired from.
Therefore, it is of paramount importance to provide equivalent characterizations
of these relations, particularly in the form of inference systems.

\begin{table}
  \[
    \begin{array}{@{}c@{}}
      \inferrule[s-nil]{\mathstrut}{
        \subt\TNil\T
      }
      \qquad
      \inferrule[s-end]{\mathstrut}{
        \subt{\End[p\,]}\T
      }
      ~T \neq \TNil
      \qquad
      \infercorule[s-converge]{
        \forall\actionsA\in\traces\T\setminus\traces\S:
        \exists\actionsB \leq \actionsA, x \in \Message:
        \subt{T(\actionsB\Out\x)}{S(\actionsB\Out\x)}
      }{
        \subt\T\S
      }
      \\\\
      \inferrule[s-inp]{
        \subt{S_x}{T_x}~{}^{(\forall x\in X)}
      }{
        \subt{\In\set{x:S_x}_{x\in X}}{\In\set{x:T_x}_{x\in Y}}
      }
      ~
      \emptyset \neq X \subseteq Y
      \qquad
      \inferrule[s-out]{
        \subt{S_x}{T_x}~{}^{(\forall x\in Y)}
      }{
        \subt{\Out\set{x:S_x}_{x\in X}}{\Out\set{x:T_x}_{x\in Y}}
      }
      ~
      \emptyset \neq Y \subseteq X
    \end{array}
  \]
  \caption{
    \label{tab:subt}
    Generalized inference system $\gis{\is[F]}{\cois[F]}$ for fair
    subtyping.  }
\end{table}

The \GIS $\gis{\is[F]}{\cois[F]}$ for fair subtyping is shown in
Table~\ref{tab:subt} and described hereafter.
Rule~\refrule{s-nil} states that $\TNil$ is the least element of the subtyping
preorder, which is justified by the fact that no client successfully interacts
with $\TNil$.
Rule~\refrule{s-end} establishes that $\End[\In]$ and $\End[\Out]$ are the least
elements among all session types different from $\TNil$. In our theory, this
relation arises from the asymmetric form of compliance we have considered: a
server $\End[p\,]$ satisfies only $\Win$, which successfully interacts with any
server different from $\TNil$.
Rules~\refrule{s-inp} and~\refrule{s-out} indicate that inputs are covariant and
outputs are contravariant. That is, it is safe to replace a server with another
one that receives a superset of messages ($X \subseteq Y$ in~\refrule{s-inp})
and, dually, it is safe to replace a server with another one that sends a subset
of messages ($Y \subseteq X$ in~\refrule{s-out}). The side condition $Y \neq
\emptyset$ in~\refrule{s-out} is important to preserve progress: if the server
that behaves according to the larger session type is unable to send any message,
the client may get stuck waiting for a message that is never sent. On the other
hand, the side condition $X \neq \emptyset$ is unnecessary from a purely
technical view point, since the rule~\refrule{s-inp} without this side condition
is subsumed by~\refrule{s-end}. We have included the side condition to minimize
the overlap between different rules and for symmetry with respect to
\refrule{s-out}.
Overall, these rules are aligned with those of the subtyping relation for
session types given by Gay and Hole~\cite{GayHole05} (see also
Remark~\ref{rem:duality}).

To discuss the corule~\refrule{s-converge} that characterizes fair subtyping we
need to introduce one last piece of notation concerning session types.

\begin{defi}[residual of a session type]
  \label{def:residual}
  Given a session type $T$ and a trace $\actions$ of $T$ we write
  $T(\actions)$ for the \emph{residual} of $T$ \emph{after}
  $\actions$, namely for the session type $S$ such that
  $T \wlred\actions S$.
\end{defi}

The notion of residual is well defined since session types are
\emph{deterministic}: if $T \wlred\actions S_1$ and $T \wlred\actions S_2$, then
$S_1 = S_2$.
It is implied that $T(\actions)$ is undefined if
$\actions\not\in\traces\T$.

For the sake of presentation we describe the corule~\refrule{s-converge}
incrementally, showing how it contributes to the soundness proof of the GIS in
Table~\ref{tab:subt}. In doing so, it helps bearing in mind that the relation
$\subt\T\S$ is meant to preserve fair compliance (Definition~\ref{def:fcomp}),
namely the possibility that any client of $\T$ can terminate successfully when
interacting with $S$.
As a first approximation observe that, when the traces of $\T$ are included in
the traces of $\S$, the corule~\refrule{s-converge} boils down to the following
coaxiom:
\[
  \infercorule{ }{
    \subt\T\S
  }
  ~\traces\T \subseteq \traces\S
\]

Now consider a client $R$ that is fair compliant with $T$. It must
be the case that $\session\R\T \wred \session\Win{T'}$ for some
$T' \neq \TNil$, namely that $\R \wlred{\co\actions} \Win$ and
$T \wlred\actions T'$ for some sequence $\actions$ of actions. The
side condition $\traces\T \subseteq \traces\S$ ensures that
$\actions$ is also a trace of $S$, therefore
$\session\R\S \wred \session\Win{S'}$ for the $S' \neq \TNil$ such
that $S \wlred\actions S'$.
In general, we know from rule~\refrule{s-out} that $S$ may perform \emph{fewer}
outputs than $T$, hence not every trace of $T$ is necessarily a trace of $S$.
Writing $\leq$ for the prefix order relation on traces, the premises
\[
  \forall\actionsA\in\traces\T\setminus\traces\S:
  \exists\actionsB \leq \actionsA, x \in \Message:
  \subt{T(\actionsB\Out\x)}{S(\actionsB\Out\x)}
\]
of~\refrule{s-converge} make sure that, for every trace $\actionsA$ of $T$ that
is not a trace of $S$, there exists a common prefix $\actionsB$ of $T$ and $S$
and an output action $\Out\x$ shared by both $T(\actionsB)$ and $S(\actionsB)$
such that the residuals of $T$ and $S$ after $\actionsB\Out\x$ are \emph{one
level closer}, in the proof tree for $\subt\T\S$, to the residuals of $T$ and
$S$ for which trace inclusion holds. The requirement that $\actionsB$ be
followed by an \emph{output} $\Out\x$ is essential, since the client $R$
\emph{must} be able to accept all the outputs of $T$.

Note that the corule is unsound in general. For instance,
$\subt{\Out0.\End}{\Out\Nat.\End}$ is derivable by~\refrule{s-converge} since
$\traces{\Out0.\End} \subseteq \traces{\Out\Nat.\End}$, but $\Out0.\End$ is
\emph{not} a subtype of $\Out\Nat.\End$.

\begin{exa}
  \label{ex:subt}
  Consider once again the session types $T_i$ and $S_i$ of Example~\ref{ex:running}.
  The (infinite) derivation
  \begin{equation}
    \label{eq:der1}
    \begin{prooftree}
      \[
        \[
          \vdots
          \justifies
          \subt{T_1}{S_1}
        \]
        \justifies
        \subt{\Out\Nat.T_1}{\Out\NatPlus.S_1}
        \using\refrule{s-out}
      \]
      \quad
      \[
        \justifies
        \subt\End\End
        \using\refrule{s-end}
      \]
      \justifies
      \subt{T_1}{S_1}
      \using\refrule{s-out}
    \end{prooftree}
  \end{equation}
  proves that $\subt{T_1}{S_1} \in \CoInductive{\is[F]}$ and the (infinite) derivation
  \begin{equation}
    \label{eq:der2}
    \begin{prooftree}
      \[
        \[
          \vdots
          \justifies
          \subt{T_2}{S_2}
        \]
        \justifies
        \subt{\Out\Nat.T_2}{\Out\NatPlus.S_2}
        \using\refrule{s-out}
      \]
      \quad
      \[
        \justifies
        \subt\End\End
        \using\refrule{s-end}
      \]
      \justifies
      \subt{T_2}{S_2}
      \using\refrule{s-inp}
    \end{prooftree}
  \end{equation}
  proves that $\subt{T_2}{S_2} \in \CoInductive{\is[F]}$.
  In order to derive $\subt{\T_i}{\S_i}$ in the \GIS $\gis{\is[F]}{\cois[F]}$ we
  must find a well-founded proof tree of $\subt\T\S$ in $\is[F] \cup \cois[F]$
  and the only hope to do so is by means of~\refrule{s-converge}, since $T_i$
  and $S_i$ share traces of arbitrary length.
  Observe that every trace $\actionsA$ of $T_1$ that is not a trace
  of $S_1$ has the form $(\Out\vtrue\Out{p_k})^k\Out\vtrue\Out0\dots$
  where $p_k \in \NatPlus$. Thus, it suffices to take
  $\actionsB = \es$ and $x = 0$, noted that
  $T_1(\Out0) = S_1(\Out0) = \End$, to derive
  \[
    \begin{prooftree}
      \[
        \Justifies
        \subt\End\End
        \using\refrule{s-converge}
      \]
      \Justifies
      \subt{T_1}{S_1}
      \using\refrule{s-converge}
    \end{prooftree}
  \]
  On the other hand, every trace $\actionsA \in \traces{T_2} \setminus
  \traces{S_2}$ has the form $(\In\vtrue\Out{p_k})^k\In\vtrue\Out0\dots$ where
  $p_k \in \NatPlus$.  All the prefixes of such traces that are followed by an
  output and are shared by both $T_2$ and $S_2$ have the form
  $(\In\vtrue\Out{p_k})^k\In\vtrue$ where $p_k \in \NatPlus$, and
  $T_2(\actionsB\Out{p}) = T_2$ and $S_2(\actionsB\Out{p}) = S_2$ for all such
  prefixes and $p \in \NatPlus$. It follows that we are unable to derive
  $\subt{T_2}{S_2}$ with a well-founded proof tree in $\is[F] \cup \cois[F]$.
  This is consistent with the fact that, in Example~\ref{ex:fcomp}, we have
  found a client $R_2$ that is fairly compliant with $T_2$ but not with $S_2$.
  Intuitively, $R_2$ insists on poking the server waiting to receive $0$. This
  may happen with $T_2$, but not with $S_2$.
  In the case of $T_1$ and $S_1$ no such client can exist, since the server may
  decide to interrupt the interaction at any time by sending a $\vfalse$ message
  to the client.
  \eoe
\end{exa}

Example~\ref{ex:subt} also shows that fair subtyping is a
\emph{context-sensitive} relation in that the applicability of a rule for
deriving $\subt\T\S$ may depend on the context in which $T$ and $S$ occur.
Indeed, in the infinite derivation~\eqref{eq:der1} the rule~\refrule{s-out} is
used infinitely many times to relate the output session types $\Out\Nat.T_1$ and
$\Out\NatPlus.S_1$. In this context, rule~\refrule{s-out} can be applied
harmlessly.
On the contrary, the infinite derivation~\eqref{eq:der2} is isomorphic to the
previous one with the difference that some applications of~\refrule{s-out} have
been replaced by applications of~\refrule{s-inp}. Here too~\refrule{s-out} is
used infinitely many times, but this time to relate the output session types
$\Out\Nat.T_2$ and $\Out\NatPlus.S_2$. This derivation allows us to prove
$\subt{T_2}{S_2} \in \CoInductive{\is[F]}$, but not $\subt{T_2}{S_2} \in
\FlexCo{\is[F]}{\cois[F]}$, because~\refrule{s-out} removes the $0$ output from
$S_2$ that a client of $T_2$ may depend upon.

\begin{rem}
  Part of the reason why rule~\refrule{s-converge} is so contrived and hard to
  understand is that the property it enforces is fundamentally \emph{non-local}
  and therefore difficult to express in terms of immediate subtrees of a session
  type.
  To better illustrate the point, consider the following alternative set of
  corules meant to replace~\refrule{s-converge} in Table~\ref{tab:subt}:
  \[
    \infercorule[co-inc]{
      \mathstrut
    }{
      \subt\T\S
    }
    ~ \traces\T \subseteq \traces\S
    \qquad
    \infercorule[co-inp]{
      \subt{S_x}{T_x}~{}^{(\forall x\in\Message)}
    }{
      \subt{\In\set{x:S_x}_{x\in\Message}}{\In\set{x:T_x}_{x\in\Message}}
    }
    \qquad
    \infercorule[co-out]{
      \subt{S}{T}
    }{
      \subt{\Out x.S + S'}{\Out x.T + T'}
    }
  \]

  It is easy to see that these rules provide a sound approximation of~\refrule{s-converge}, but they are not complete. Indeed, consider the session
  types $T = \In\vtrue.T \branch \In\vfalse.(\Out\vtrue.\End +
  \Out\vfalse.\End)$ and $S = \In\vtrue.S \branch \In\vfalse.\Out\vtrue.\End$.
  We have $\subt\T\S$ and yet $\isubt\T\S$ cannot be proved with the above corules: it is not possible to prove $\isubt\T\S$ using~\refrule{co-inc} because $\traces\T \not\subseteq \traces\S$. If, on the other hand, we insist on visiting both branches of the topmost input as required by~\refrule{co-inp}, we end up requiring a proof of $\isubt\T\S$ in order to derive $\isubt\T\S$.
  \eor
\end{rem}

\begin{thm}
  \label{thm:sub}
  For every $T, S \in \SessionType$ the following properties hold:
  \begin{enumerate}
  \item $T$ is a subtype of $S$ if and only if
    $\subt\T\S \in \CoInductive{\is[F]}$;
  \item $T$ is a fair subtype of $S$ if and only if
    $\subt\T\S \in \FlexCo{\is[F]}{\cois[F]}$.
  \end{enumerate}
\end{thm}
\begin{proof}[Proof sketch]
  As usual we focus on item~(2), which is the most interesting property.
  For the ``if'' part, we consider an arbitrary $R$ that fairly complies with
  $T$ and show that it fairly complies with $S$ as well. More specifically, we
  consider a reduction $\session{R}{S} \wred \session{R'}{S'}$ and show that it
  can be extended so as to achieve client satisfaction.
  The first step is to ``unzip'' this reduction into $R \wlred{\co\actions} R'$
  and $S \wlred\actions S'$ for some string $\actions$ of actions. Then, we show
  by induction on $\actions$ that there exists $T'$ such that $\subt{T'}{S'} \in
  \FlexCo{\is[F]}{\cois[F]}$ and $T \wlred\actions T'$, using the hypothesis
  $\subt{T}{S} \in \CoInductive{\is[F]}$ and the hypothesis that $R$ complies
  with $T$. This means that $R$ and $T$ may synchronize just like $R$ and $S$,
  obtaining a reduction $\session{R}{T} \wred \session{R'}{T'}$. At this point
  the existence of the reduction $\session{R'}{S'} \wred \session\Win{S''}$ is
  proved using the arguments in the discussion of rule~\refrule{s-converge}
  given earlier.

  For the ``only if'' part we use once again the bounded coinduction principle.
  In particular, we use the hypothesis that $T$ is a fair subtype of $S$ to show
  that $T$ and $S$ must have one of the forms in the conclusions of the rules in
  Table~\ref{tab:subt}. This proof is done by cases on the shape of $T$,
  constructing a canonical client of $T$ that must succeed with $S$ as well.
  Then, the coinduction principle allows us to conclude that $\subt{T}{S} \in
  \CoInductive{\is[F]}$.
  The fact that $\subt{T}{S} \in \Inductive{\is[F] \cup \cois[F]}$ also holds is
  by far the most intricate step of the proof. First of all, we establish that
  $\Inductive{\is[F] \cup \cois[F]} = \Inductive{\cois[F]}$. That is, we
  establish that rule~\refrule{s-converge} subsumes all the rules in $\is[F]$
  when they are inductively interpreted. Then, we provide a characterization of
  the \emph{negation} of~\refrule{s-converge}, which we call divergence. At this
  point we proceed by contradiction: under the hypothesis that $\subt{T}{S} \in
  \CoInductive{\is[F]}$ and that $T$ and $S$ ``diverge'', we are able to
  corecursively define a \emph{discriminating} client $R$ that fairly complies
  with $T$ but not with $S$. This contradicts the hypothesis that $T$ is a fair
  subtype of $S$ and proves, albeit in a non-constructive way, that $\subt{T}{S}
  \in \Inductive{\cois[F]}$ as requested.
\end{proof}

\begin{rem}
  \label{rem:duality}
  Most session type theories adopt a symmetric form of session type
  compatibility whereby client and server are required to terminate
  the interaction at the same time.
  It is easy to define a notion of \emph{symmetric compliance} (also known as
  \emph{peer compliance}~\cite{BernardiHennessy16}) by turning $T' \ne \TNil$
  into $T' = \End$ in Definition~\ref{def:comp}. The subtyping relation induced
  by symmetric compliance has essentially the same characterization of
  Definition~\ref{def:sub}, except that the axiom~\refrule{s-end} is replaced by
  the more familiar $\subt{\End[p\,]}{\End[q\,]}$~\cite{GayHole05}.
  On the other hand, the analogous change in Definition~\ref{def:fcomp} has much
  deeper consequences: the requirement that client and server must end the
  interaction at the same time induces a large family of session types that are
  syntactically very different, but semantically equivalent. For example, the
  session types $T$ and $S$ such that $T = \In\Nat.T$ and $S = \Out\Bool.S$,
  which describe completely unrelated protocols, would be equivalent for the
  simple reason that no client successfully interacts with them (they are not
  fairly terminating, since they do not contain any occurrence of $\End[]$).
  We have not investigated the existence of a \GIS for fair subtyping induced by
  symmetric fair compliance. A partial characterization (which however requires
  various auxiliary relations) is given by Padovani~\cite{Padovani16}.
  \eor
\end{rem}


%% file: agda.tex
\newcommand{\linkrepo}[1]{%
  \href{https://github.com/boystrange/FairSubtypingAgda/blob/main/src/#1}{\texttt{#1}}%
}
\renewcommand{\AgdaFontStyle}[1]{\texttt{#1}}
\renewcommand{\AgdaKeywordFontStyle}[1]{\texttt{#1}}
\renewcommand{\AgdaBoundFontStyle}[1]{\texttt{#1}}
\renewcommand{\AgdaUnderscore}{\char`_}

\begin{code}[hide]%
\>[0]\AgdaSymbol{\{-\#}\AgdaSpace{}%
\AgdaKeyword{OPTIONS}\AgdaSpace{}%
\AgdaPragma{--guardedness}\AgdaSpace{}%
\AgdaPragma{--sized-types}\AgdaSpace{}%
\AgdaSymbol{\#-\}}\<%
\\
\\[\AgdaEmptyExtraSkip]%
\>[0]\AgdaKeyword{import}\AgdaSpace{}%
\AgdaModule{Level}\<%
\\
\>[0]\AgdaKeyword{open}\AgdaSpace{}%
\AgdaKeyword{import}\AgdaSpace{}%
\AgdaModule{Data.Bool}\AgdaSpace{}%
\AgdaKeyword{using}\AgdaSpace{}%
\AgdaSymbol{(}\AgdaDatatype{Bool}\AgdaSymbol{;}\AgdaSpace{}%
\AgdaInductiveConstructor{true}\AgdaSymbol{;}\AgdaSpace{}%
\AgdaInductiveConstructor{false}\AgdaSymbol{)}\<%
\\
\>[0]\AgdaKeyword{import}\AgdaSpace{}%
\AgdaModule{Data.Fin}\AgdaSpace{}%
\AgdaKeyword{using}\AgdaSpace{}%
\AgdaSymbol{(}\AgdaInductiveConstructor{zero}\AgdaSymbol{;}\AgdaSpace{}%
\AgdaInductiveConstructor{suc}\AgdaSymbol{)}\<%
\\
\>[0]\AgdaKeyword{open}\AgdaSpace{}%
\AgdaKeyword{import}\AgdaSpace{}%
\AgdaModule{Data.Empty}\AgdaSpace{}%
\AgdaKeyword{using}\AgdaSpace{}%
\AgdaSymbol{(}\AgdaDatatype{⊥}\AgdaSymbol{;}\AgdaSpace{}%
\AgdaFunction{⊥-elim}\AgdaSymbol{)}\<%
\\
\>[0]\AgdaKeyword{open}\AgdaSpace{}%
\AgdaKeyword{import}\AgdaSpace{}%
\AgdaModule{Data.Nat}\AgdaSpace{}%
\AgdaKeyword{using}\AgdaSpace{}%
\AgdaSymbol{(}\AgdaDatatype{ℕ}\AgdaSymbol{;}\AgdaSpace{}%
\AgdaInductiveConstructor{zero}\AgdaSymbol{;}\AgdaSpace{}%
\AgdaInductiveConstructor{suc}\AgdaSymbol{)}\<%
\\
\>[0]\AgdaKeyword{open}\AgdaSpace{}%
\AgdaKeyword{import}\AgdaSpace{}%
\AgdaModule{Data.Vec}\AgdaSpace{}%
\AgdaKeyword{using}\AgdaSpace{}%
\AgdaSymbol{(}\AgdaInductiveConstructor{[]}\AgdaSymbol{;}\AgdaSpace{}%
\AgdaOperator{\AgdaInductiveConstructor{\AgdaUnderscore{}∷\AgdaUnderscore{}}}\AgdaSymbol{)}\<%
\\
\>[0]\AgdaKeyword{open}\AgdaSpace{}%
\AgdaKeyword{import}\AgdaSpace{}%
\AgdaModule{Data.Sum}\AgdaSpace{}%
\AgdaKeyword{using}\AgdaSpace{}%
\AgdaSymbol{(}\AgdaInductiveConstructor{inj₁}\AgdaSymbol{;}\AgdaSpace{}%
\AgdaInductiveConstructor{inj₂}\AgdaSymbol{)}\<%
\\
\>[0]\AgdaKeyword{open}\AgdaSpace{}%
\AgdaKeyword{import}\AgdaSpace{}%
\AgdaModule{Data.Product}\<%
\\
\>[0]\AgdaKeyword{open}\AgdaSpace{}%
\AgdaKeyword{import}\AgdaSpace{}%
\AgdaModule{Relation.Unary}\AgdaSpace{}%
\AgdaKeyword{using}\AgdaSpace{}%
\AgdaSymbol{(}\AgdaFunction{Pred}\AgdaSymbol{;}\AgdaSpace{}%
\AgdaOperator{\AgdaFunction{\AgdaUnderscore{}∈\AgdaUnderscore{}}}\AgdaSymbol{;}\AgdaSpace{}%
\AgdaOperator{\AgdaFunction{\AgdaUnderscore{}⊆\AgdaUnderscore{}}}\AgdaSymbol{;}\AgdaSpace{}%
\AgdaOperator{\AgdaFunction{\AgdaUnderscore{}∩\AgdaUnderscore{}}}\AgdaSymbol{;}\AgdaSpace{}%
\AgdaFunction{Empty}\AgdaSymbol{;}\AgdaSpace{}%
\AgdaFunction{Satisfiable}\AgdaSymbol{)}\<%
\\
\>[0]\AgdaKeyword{open}\AgdaSpace{}%
\AgdaKeyword{import}\AgdaSpace{}%
\AgdaModule{Relation.Binary.Construct.Closure.ReflexiveTransitive}\AgdaSpace{}%
\AgdaKeyword{using}\AgdaSpace{}%
\AgdaSymbol{(}\AgdaDatatype{Star}\AgdaSymbol{;}\AgdaSpace{}%
\AgdaInductiveConstructor{ε}\AgdaSymbol{;}\AgdaSpace{}%
\AgdaOperator{\AgdaInductiveConstructor{\AgdaUnderscore{}◅\AgdaUnderscore{}}}\AgdaSymbol{)}\<%
\\
\>[0]\AgdaKeyword{open}\AgdaSpace{}%
\AgdaKeyword{import}\AgdaSpace{}%
\AgdaModule{Relation.Binary.PropositionalEquality}\AgdaSpace{}%
\AgdaKeyword{using}\AgdaSpace{}%
\AgdaSymbol{(}\AgdaInductiveConstructor{refl}\AgdaSymbol{)}\<%
\\
\\[\AgdaEmptyExtraSkip]%
\>[0]\AgdaKeyword{open}\AgdaSpace{}%
\AgdaKeyword{import}\AgdaSpace{}%
\AgdaModule{is-lib.InfSys}\AgdaSpace{}%
\AgdaKeyword{renaming}\AgdaSpace{}%
\AgdaSymbol{(}\AgdaOperator{\AgdaFunction{FCoInd⟦\AgdaUnderscore{},\AgdaUnderscore{}⟧}}\AgdaSpace{}%
\AgdaSymbol{to}\AgdaSpace{}%
\AgdaOperator{\AgdaFunction{Gen⟦\AgdaUnderscore{},\AgdaUnderscore{}⟧}}\AgdaSymbol{)}\<%
\\
\\[\AgdaEmptyExtraSkip]%
\>[0]\AgdaComment{--\ postulate\ 𝕍\ :\ Set}\<%
\\
\>[0]\AgdaKeyword{data}\AgdaSpace{}%
\AgdaDatatype{𝕍}\AgdaSpace{}%
\AgdaSymbol{:}\AgdaSpace{}%
\AgdaPrimitive{Set}\AgdaSpace{}%
\AgdaKeyword{where}\<%
\\
\>[0][@{}l@{\AgdaIndent{0}}]%
\>[2]\AgdaInductiveConstructor{nat}%
\>[8]\AgdaSymbol{:}\AgdaSpace{}%
\AgdaDatatype{ℕ}\AgdaSpace{}%
\AgdaSymbol{→}\AgdaSpace{}%
\AgdaDatatype{𝕍}\<%
\\
\>[2]\AgdaInductiveConstructor{bool}%
\>[8]\AgdaSymbol{:}\AgdaSpace{}%
\AgdaDatatype{Bool}\AgdaSpace{}%
\AgdaSymbol{→}\AgdaSpace{}%
\AgdaDatatype{𝕍}\<%
\end{code}

\section{Agda Formalization of Fair Compliance}
\label{sec:agda}

In this section we provide a walkthrough of the Agda formalization
of fair compliance (Section~\ref{sec:compliance}). Among all the
properties we have formalized, we focus on fair compliance because
it is sufficiently representative but also accessible. The
formalization of fair termination is simpler, whereas that of fair
subtyping is substantially more involved. In this section we assume
that the reader has some acquaintance with Agda.  The
code presented here and the full development~\cite{CicconePadovani21} have been
checked using Agda 2.6.2.1 and \texttt{agda-stdlib} 1.7.1.

\subsection{Representation of session types}

We postulate the existence of $\Message$, representing the set of
values that can be exchanged in communications.

\medskip
\noindent
\begin{AgdaAlign}%
\AgdaKeyword{postulate} \AgdaDatatype{𝕍} : \AgdaDatatype{Set}
\end{AgdaAlign}
\medskip

The actual Agda formalization is parametric on an
arbitrary set $\Message$, the only requirement being that $\Message$
must be equipped with a decidable notion of equality. We will make
some assumptions on the nature of $\Message$ when we present
specific examples.
To begin the formalization, we declare the two data types we use to represent session types. Because these types are mutually recursive, we declare them in advance so that we can later refer to them from within the definition of each.
\begin{code}%
\>[0]\AgdaKeyword{data}%
\>[9]\AgdaDatatype{SessionType}\AgdaSpace{}%
\AgdaSymbol{:}\AgdaSpace{}%
\AgdaPrimitive{Set}\<%
\\
\>[0]\AgdaKeyword{record}%
\>[8]\AgdaRecord{∞SessionType}\AgdaSpace{}%
\AgdaSymbol{:}\AgdaSpace{}%
\AgdaPrimitive{Set}\<%
\end{code}

A key design choice of our formalization is the representation of continuations $\set{x:S_x}_{x\in\Message}$, which may have infinitely many branches if $\Message$ is infinite. We represent continuations in Agda as \emph{total functions} from $\Message$ to session types, thus:
\begin{code}%
\>[0]\AgdaFunction{Continuation}\AgdaSpace{}%
\AgdaSymbol{:}\AgdaSpace{}%
\AgdaPrimitive{Set}\<%
\\
\>[0]\AgdaFunction{Continuation}\AgdaSpace{}%
\AgdaSymbol{=}\AgdaSpace{}%
\AgdaDatatype{𝕍}\AgdaSpace{}%
\AgdaSymbol{→}\AgdaSpace{}%
\AgdaRecord{∞SessionType}\<%
\end{code}

The $\AgdaDatatype{SessionType}$ data type provides three constructors corresponding to the three forms of a session type (Definition~\ref{def:st}):
\begin{code}%
\>[0]\AgdaKeyword{data}\AgdaSpace{}%
\AgdaDatatype{SessionType}\AgdaSpace{}%
\AgdaKeyword{where}\<%
\\
\>[0][@{}l@{\AgdaIndent{0}}]%
\>[2]\AgdaInductiveConstructor{nil}%
\>[11]\AgdaSymbol{:}\AgdaSpace{}%
\AgdaDatatype{SessionType}\<%
\\
\>[2]\AgdaInductiveConstructor{inp}\AgdaSpace{}%
\AgdaInductiveConstructor{out}%
\>[11]\AgdaSymbol{:}\AgdaSpace{}%
\AgdaFunction{Continuation}\AgdaSpace{}%
\AgdaSymbol{→}\AgdaSpace{}%
\AgdaDatatype{SessionType}\<%
\end{code}

The $\AgdaDatatype{∞SessionType}$ wraps $\AgdaDatatype{SessionType}$
within a coinductive record, so as to make it possible to represent
\emph{infinite} session types. The record has just one field
$\AgdaField{force}$ that can be used to access the wrapped session
type. By opening the record, we make the $\AgdaField{force}$ field
publicly accessible without qualifiers.
\begin{code}%
\>[0]\AgdaKeyword{record}\AgdaSpace{}%
\AgdaRecord{∞SessionType}\AgdaSpace{}%
\AgdaKeyword{where}\<%
\\
\>[0][@{}l@{\AgdaIndent{0}}]%
\>[2]\AgdaKeyword{coinductive}\<%
\\
\>[2]\AgdaKeyword{field}\AgdaSpace{}%
\AgdaField{force}\AgdaSpace{}%
\AgdaSymbol{:}\AgdaSpace{}%
\AgdaDatatype{SessionType}\<%
\\
\>[0]\AgdaKeyword{open}\AgdaSpace{}%
\AgdaModule{∞SessionType}\AgdaSpace{}%
\AgdaKeyword{public}\<%
\end{code}

As an example, consider once again the session type $T_1$ discussed in Example~\ref{ex:running}:
\[
  T_1 = \Out\vtrue.\Out\Nat.\T_1 + \Out\vfalse.\End
\]
In this case we assume that $\Message$ is the disjoint sum of $\Nat$
(the set of natural numbers) and $\Bool$ (the set of boolean values)
with constructors $\AgdaInductiveConstructor{nat}$ and
$\AgdaInductiveConstructor{bool}$. It is useful to also define once
and for all the \emph{continuation} $\AgdaFunction{empty}$, which
maps every message to $\TNil$:
\begin{code}%
\>[0]\AgdaFunction{empty}\AgdaSpace{}%
\AgdaSymbol{:}\AgdaSpace{}%
\AgdaFunction{Continuation}\<%
\\
\>[0]\AgdaFunction{empty}\AgdaSpace{}%
\AgdaSymbol{\AgdaUnderscore{}}\AgdaSpace{}%
\AgdaSymbol{.}\AgdaField{force}\AgdaSpace{}%
\AgdaSymbol{=}\AgdaSpace{}%
\AgdaInductiveConstructor{nil}\<%
\end{code}

Now, the Agda encoding of $T_1$ is shown below:

\begin{code}%
\>[0]\AgdaFunction{T₁}\AgdaSpace{}%
\AgdaSymbol{:}\AgdaSpace{}%
\AgdaDatatype{SessionType}\<%
\\
\>[0]\AgdaFunction{T₁}\AgdaSpace{}%
\AgdaSymbol{=}\AgdaSpace{}%
\AgdaInductiveConstructor{out}\AgdaSpace{}%
\AgdaFunction{f}\<%
\\
\>[0][@{}l@{\AgdaIndent{0}}]%
\>[2]\AgdaKeyword{where}\<%
\\
\>[2][@{}l@{\AgdaIndent{0}}]%
\>[4]\AgdaFunction{f}\AgdaSpace{}%
\AgdaFunction{g}\AgdaSpace{}%
\AgdaSymbol{:}\AgdaSpace{}%
\AgdaFunction{Continuation}\<%
\\
\\[\AgdaEmptyExtraSkip]%
\>[4]\AgdaFunction{f}\AgdaSpace{}%
\AgdaSymbol{(}\AgdaInductiveConstructor{nat}\AgdaSpace{}%
\AgdaSymbol{\AgdaUnderscore{})}%
\>[20]\AgdaSymbol{.}\AgdaField{force}\AgdaSpace{}%
\AgdaSymbol{=}\AgdaSpace{}%
\AgdaInductiveConstructor{nil}\<%
\\
\>[4]\AgdaFunction{f}\AgdaSpace{}%
\AgdaSymbol{(}\AgdaInductiveConstructor{bool}\AgdaSpace{}%
\AgdaInductiveConstructor{true}\AgdaSymbol{)}%
\>[20]\AgdaSymbol{.}\AgdaField{force}\AgdaSpace{}%
\AgdaSymbol{=}\AgdaSpace{}%
\AgdaInductiveConstructor{out}\AgdaSpace{}%
\AgdaFunction{g}\<%
\\
\>[4]\AgdaFunction{f}\AgdaSpace{}%
\AgdaSymbol{(}\AgdaInductiveConstructor{bool}\AgdaSpace{}%
\AgdaInductiveConstructor{false}\AgdaSymbol{)}%
\>[20]\AgdaSymbol{.}\AgdaField{force}\AgdaSpace{}%
\AgdaSymbol{=}\AgdaSpace{}%
\AgdaInductiveConstructor{inp}\AgdaSpace{}%
\AgdaFunction{empty}\<%
\\
\\[\AgdaEmptyExtraSkip]%
\>[4]\AgdaFunction{g}\AgdaSpace{}%
\AgdaSymbol{(}\AgdaInductiveConstructor{nat}\AgdaSpace{}%
\AgdaSymbol{\AgdaUnderscore{})}%
\>[20]\AgdaSymbol{.}\AgdaField{force}\AgdaSpace{}%
\AgdaSymbol{=}\AgdaSpace{}%
\AgdaInductiveConstructor{out}\AgdaSpace{}%
\AgdaFunction{f}\<%
\\
\>[4]\AgdaFunction{g}\AgdaSpace{}%
\AgdaSymbol{(}\AgdaInductiveConstructor{bool}\AgdaSpace{}%
\AgdaSymbol{\AgdaUnderscore{})}%
\>[20]\AgdaSymbol{.}\AgdaField{force}\AgdaSpace{}%
\AgdaSymbol{=}\AgdaSpace{}%
\AgdaInductiveConstructor{nil}\<%
\end{code}

The continuations $\AgdaFunction{f}$ and
$\AgdaFunction{g}$ are defined using pattern matching on the
message argument and using copattern matching to specify the value of the $\AgdaField{force}$ field of the resulting coinductive record. They represent the two stages of the protocol:
$\AgdaFunction{f}$ allows sending a boolean (but no natural number)
and, depending on the boolean, it continues as $\AgdaFunction{g}$ or
it terminates; $\AgdaFunction{g}$ allows sending a natural number
(but no boolean) and continues as $\AgdaFunction{f}$.
This example illustrates a simple form of \emph{dependency} whereby
the structure of a communication protocol may depend on the content
of previously exchanged messages. The fact that we use Agda to write
continuations means that we can model sophisticated forms of
dependencies that are found only in the most advanced theories of
dependent session types~\cite{ToninhoCairesPfenning11,ToninhoYoshida18,ThiemannVasconcelos20,CicconePadovani20}. For
example, below is the Agda encoding of a session type
\[
  \Out{(n:\Nat)}.\underbrace{\Out\Bool\dots\Out\Bool}_n.\End[\Out]
\]
describing a channel used for sending a natural number $n$ followed by $n$ boolean values:
\begin{code}%
\>[0]\AgdaFunction{BoolVector}\AgdaSpace{}%
\AgdaSymbol{:}\AgdaSpace{}%
\AgdaDatatype{SessionType}\<%
\\
\>[0]\AgdaFunction{BoolVector}\AgdaSpace{}%
\AgdaSymbol{=}\AgdaSpace{}%
\AgdaInductiveConstructor{out}\AgdaSpace{}%
\AgdaFunction{g}\<%
\\
\>[0][@{}l@{\AgdaIndent{0}}]%
\>[2]\AgdaKeyword{where}\<%
\\
\>[2][@{}l@{\AgdaIndent{0}}]%
\>[4]\AgdaFunction{f}\AgdaSpace{}%
\AgdaSymbol{:}\AgdaSpace{}%
\AgdaDatatype{ℕ}\AgdaSpace{}%
\AgdaSymbol{→}\AgdaSpace{}%
\AgdaFunction{Continuation}\<%
\\
\>[4]\AgdaFunction{f}\AgdaSpace{}%
\AgdaInductiveConstructor{zero}%
\>[15]\AgdaSymbol{\AgdaUnderscore{}}%
\>[26]\AgdaSymbol{.}\AgdaField{force}\AgdaSpace{}%
\AgdaSymbol{=}\AgdaSpace{}%
\AgdaInductiveConstructor{nil}\<%
\\
\>[4]\AgdaFunction{f}\AgdaSpace{}%
\AgdaSymbol{(}\AgdaInductiveConstructor{suc}\AgdaSpace{}%
\AgdaBound{n}\AgdaSymbol{)}%
\>[15]\AgdaSymbol{(}\AgdaInductiveConstructor{bool}%
\>[22]\AgdaSymbol{\AgdaUnderscore{})}%
\>[26]\AgdaSymbol{.}\AgdaField{force}\AgdaSpace{}%
\AgdaSymbol{=}\AgdaSpace{}%
\AgdaInductiveConstructor{out}\AgdaSpace{}%
\AgdaSymbol{(}\AgdaFunction{f}\AgdaSpace{}%
\AgdaBound{n}\AgdaSymbol{)}\<%
\\
\>[4]\AgdaFunction{f}\AgdaSpace{}%
\AgdaSymbol{(}\AgdaInductiveConstructor{suc}\AgdaSpace{}%
\AgdaBound{n}\AgdaSymbol{)}%
\>[15]\AgdaSymbol{(}\AgdaInductiveConstructor{nat}%
\>[22]\AgdaSymbol{\AgdaUnderscore{})}%
\>[26]\AgdaSymbol{.}\AgdaField{force}\AgdaSpace{}%
\AgdaSymbol{=}\AgdaSpace{}%
\AgdaInductiveConstructor{nil}\<%
\\
\\[\AgdaEmptyExtraSkip]%
\>[4]\AgdaFunction{g}\AgdaSpace{}%
\AgdaSymbol{:}\AgdaSpace{}%
\AgdaFunction{Continuation}\<%
\\
\>[4]\AgdaFunction{g}\AgdaSpace{}%
\AgdaSymbol{(}\AgdaInductiveConstructor{nat}%
\>[13]\AgdaBound{n}\AgdaSymbol{)}%
\>[26]\AgdaSymbol{.}\AgdaField{force}\AgdaSpace{}%
\AgdaSymbol{=}\AgdaSpace{}%
\AgdaInductiveConstructor{out}\AgdaSpace{}%
\AgdaSymbol{(}\AgdaFunction{f}\AgdaSpace{}%
\AgdaBound{n}\AgdaSymbol{)}\<%
\\
\>[4]\AgdaFunction{g}\AgdaSpace{}%
\AgdaSymbol{(}\AgdaInductiveConstructor{bool}%
\>[13]\AgdaSymbol{\AgdaUnderscore{})}%
\>[26]\AgdaSymbol{.}\AgdaField{force}\AgdaSpace{}%
\AgdaSymbol{=}\AgdaSpace{}%
\AgdaInductiveConstructor{nil}\<%
\end{code}

We will not discuss further examples of dependent session types. However, note that the possibility of encoding protocols such as $\AgdaFunction{BoolVector}$ has important implications on the scope of our study: it means that the results we have presented and formally proved in Agda hold for a large family of session types that includes dependent ones.

We now provide a few auxiliary predicates on session types and continuations. First of all, we say that a session type is \emph{defined} if it is different from $\TNil$:
\begin{code}%
\>[0]\AgdaKeyword{data}\AgdaSpace{}%
\AgdaDatatype{Defined}\AgdaSpace{}%
\AgdaSymbol{:}\AgdaSpace{}%
\AgdaDatatype{SessionType}\AgdaSpace{}%
\AgdaSymbol{→}\AgdaSpace{}%
\AgdaPrimitive{Set}\AgdaSpace{}%
\AgdaKeyword{where}\<%
\\
\>[0][@{}l@{\AgdaIndent{0}}]%
\>[2]\AgdaInductiveConstructor{inp}%
\>[7]\AgdaSymbol{:}\AgdaSpace{}%
\AgdaSymbol{∀\{}\AgdaBound{f}\AgdaSymbol{\}}\AgdaSpace{}%
\AgdaSymbol{→}\AgdaSpace{}%
\AgdaDatatype{Defined}\AgdaSpace{}%
\AgdaSymbol{(}\AgdaInductiveConstructor{inp}\AgdaSpace{}%
\AgdaBound{f}\AgdaSymbol{)}\<%
\\
\>[2]\AgdaInductiveConstructor{out}%
\>[7]\AgdaSymbol{:}\AgdaSpace{}%
\AgdaSymbol{∀\{}\AgdaBound{f}\AgdaSymbol{\}}\AgdaSpace{}%
\AgdaSymbol{→}\AgdaSpace{}%
\AgdaDatatype{Defined}\AgdaSpace{}%
\AgdaSymbol{(}\AgdaInductiveConstructor{out}\AgdaSpace{}%
\AgdaBound{f}\AgdaSymbol{)}\<%
\end{code}

Concerning continuations, we define the \emph{domain} of a continuation function $f$ to be the subset $\AgdaFunction{dom}~f$ of $\Message$ such that $f~x$ is defined, that is $\AgdaFunction{dom}~f = \set{ x \in \Message \mid f~x \ne \TNil }$:
\begin{code}%
\>[0]\AgdaFunction{dom}\AgdaSpace{}%
\AgdaSymbol{:}\AgdaSpace{}%
\AgdaFunction{Continuation}\AgdaSpace{}%
\AgdaSymbol{→}\AgdaSpace{}%
\AgdaFunction{Pred}\AgdaSpace{}%
\AgdaDatatype{𝕍}\AgdaSpace{}%
\AgdaPrimitive{Level.zero}\<%
\\
\>[0]\AgdaFunction{dom}\AgdaSpace{}%
\AgdaBound{f}\AgdaSpace{}%
\AgdaBound{x}\AgdaSpace{}%
\AgdaSymbol{=}\AgdaSpace{}%
\AgdaDatatype{Defined}\AgdaSpace{}%
\AgdaSymbol{(}\AgdaBound{f}\AgdaSpace{}%
\AgdaBound{x}\AgdaSpace{}%
\AgdaSymbol{.}\AgdaField{force}\AgdaSymbol{)}\<%
\end{code}

A continuation function is said to be \emph{empty} if so is its domain.
\begin{code}%
\>[0]\AgdaFunction{EmptyContinuation}\AgdaSpace{}%
\AgdaSymbol{:}\AgdaSpace{}%
\AgdaFunction{Continuation}\AgdaSpace{}%
\AgdaSymbol{→}\AgdaSpace{}%
\AgdaPrimitive{Set}\<%
\\
\>[0]\AgdaFunction{EmptyContinuation}\AgdaSpace{}%
\AgdaBound{f}\AgdaSpace{}%
\AgdaSymbol{=}\AgdaSpace{}%
\AgdaFunction{Relation.Unary.Empty}\AgdaSpace{}%
\AgdaSymbol{(}\AgdaFunction{dom}\AgdaSpace{}%
\AgdaBound{f}\AgdaSymbol{)}\<%
\end{code}

In particular, the previously defined $\AgdaFunction{empty}$ continuation is indeed an empty one.
\begin{code}%
\>[0]\AgdaFunction{empty-is-empty}\AgdaSpace{}%
\AgdaSymbol{:}\AgdaSpace{}%
\AgdaFunction{EmptyContinuation}\AgdaSpace{}%
\AgdaFunction{empty}\<%
\\
\>[0]\AgdaFunction{empty-is-empty}\AgdaSpace{}%
\AgdaSymbol{\AgdaUnderscore{}}\AgdaSpace{}%
\AgdaSymbol{()}\<%
\end{code}

On the contrary, a non-empty continuation is said to have a \emph{witness}. We define a $\AgdaFunction{Witness}$ predicate to characterize this condition.
\begin{code}%
\>[0]\AgdaFunction{Witness}\AgdaSpace{}%
\AgdaSymbol{:}\AgdaSpace{}%
\AgdaFunction{Continuation}\AgdaSpace{}%
\AgdaSymbol{→}\AgdaSpace{}%
\AgdaPrimitive{Set}\<%
\\
\>[0]\AgdaFunction{Witness}\AgdaSpace{}%
\AgdaBound{f}\AgdaSpace{}%
\AgdaSymbol{=}\AgdaSpace{}%
\AgdaFunction{Relation.Unary.Satisfiable}\AgdaSpace{}%
\AgdaSymbol{(}\AgdaFunction{dom}\AgdaSpace{}%
\AgdaBound{f}\AgdaSymbol{)}\<%
\end{code}

We now define a predicate $\AgdaDatatype{Win}$ to characterize the session type $\End[\Out]$.
\begin{code}%
\>[0]\AgdaKeyword{data}\AgdaSpace{}%
\AgdaDatatype{Win}\AgdaSpace{}%
\AgdaSymbol{:}\AgdaSpace{}%
\AgdaDatatype{SessionType}\AgdaSpace{}%
\AgdaSymbol{→}\AgdaSpace{}%
\AgdaPrimitive{Set}\AgdaSpace{}%
\AgdaKeyword{where}\<%
\\
\>[0][@{}l@{\AgdaIndent{0}}]%
\>[2]\AgdaInductiveConstructor{out}\AgdaSpace{}%
\AgdaSymbol{:}\AgdaSpace{}%
\AgdaSymbol{∀\{}\AgdaBound{f}\AgdaSymbol{\}}\AgdaSpace{}%
\AgdaSymbol{→}\AgdaSpace{}%
\AgdaFunction{EmptyContinuation}\AgdaSpace{}%
\AgdaBound{f}\AgdaSpace{}%
\AgdaSymbol{→}\AgdaSpace{}%
\AgdaDatatype{Win}\AgdaSpace{}%
\AgdaSymbol{(}\AgdaInductiveConstructor{out}\AgdaSpace{}%
\AgdaBound{f}\AgdaSymbol{)}\<%
\end{code}

\subsection{Labeled transition system for session types}

Let us move onto the definition of transitions for session types. We begin by defining an $\AgdaDatatype{Action}$ data type to represent input/output actions, which consist of a polarity and a value.
\begin{code}%
\>[0]\AgdaKeyword{data}\AgdaSpace{}%
\AgdaDatatype{Action}\AgdaSpace{}%
\AgdaSymbol{:}\AgdaSpace{}%
\AgdaPrimitive{Set}\AgdaSpace{}%
\AgdaKeyword{where}\<%
\\
\>[0][@{}l@{\AgdaIndent{0}}]%
\>[2]\AgdaInductiveConstructor{I}\AgdaSpace{}%
\AgdaInductiveConstructor{O}\AgdaSpace{}%
\AgdaSymbol{:}\AgdaSpace{}%
\AgdaDatatype{𝕍}\AgdaSpace{}%
\AgdaSymbol{→}\AgdaSpace{}%
\AgdaDatatype{Action}\<%
\end{code}

The \emph{complementary action} of $\action$, denoted by $\co\action$ in the previous sections, is computed by the function $\AgdaFunction{co-action}$.
\begin{code}%
\>[0]\AgdaFunction{co-action}\AgdaSpace{}%
\AgdaSymbol{:}\AgdaSpace{}%
\AgdaDatatype{Action}\AgdaSpace{}%
\AgdaSymbol{→}\AgdaSpace{}%
\AgdaDatatype{Action}\<%
\\
\>[0]\AgdaFunction{co-action}\AgdaSpace{}%
\AgdaSymbol{(}\AgdaInductiveConstructor{I}\AgdaSpace{}%
\AgdaBound{x}\AgdaSymbol{)}%
\>[17]\AgdaSymbol{=}\AgdaSpace{}%
\AgdaInductiveConstructor{O}\AgdaSpace{}%
\AgdaBound{x}\<%
\\
\>[0]\AgdaFunction{co-action}\AgdaSpace{}%
\AgdaSymbol{(}\AgdaInductiveConstructor{O}\AgdaSpace{}%
\AgdaBound{x}\AgdaSymbol{)}%
\>[17]\AgdaSymbol{=}\AgdaSpace{}%
\AgdaInductiveConstructor{I}\AgdaSpace{}%
\AgdaBound{x}\<%
\end{code}

A \emph{transition} is a ternary relation among two session types and an action. A value of type $\AgdaDatatype{Transition}~S~\action~T$ represents the transition $S \lred\action T$. Note that the premise $x \in \AgdaFunction{dom}~f$ in the constructor $\AgdaInductiveConstructor{out}$ corresponds to the side condition $S \ne \TNil$ of rule~\refrule{output} in~\eqref{eq:lts}.
\begin{code}%
\>[0]\AgdaKeyword{data}\AgdaSpace{}%
\AgdaDatatype{Transition}\AgdaSpace{}%
\AgdaSymbol{:}\AgdaSpace{}%
\AgdaDatatype{SessionType}\AgdaSpace{}%
\AgdaSymbol{→}\AgdaSpace{}%
\AgdaDatatype{Action}\AgdaSpace{}%
\AgdaSymbol{→}\AgdaSpace{}%
\AgdaDatatype{SessionType}\AgdaSpace{}%
\AgdaSymbol{→}\AgdaSpace{}%
\AgdaPrimitive{Set}\AgdaSpace{}%
\AgdaKeyword{where}\<%
\\
\>[0][@{}l@{\AgdaIndent{0}}]%
\>[2]\AgdaInductiveConstructor{inp}%
\>[7]\AgdaSymbol{:}\AgdaSpace{}%
\AgdaSymbol{∀\{}\AgdaBound{f}\AgdaSpace{}%
\AgdaBound{x}\AgdaSymbol{\}}%
\>[29]\AgdaSymbol{→}\AgdaSpace{}%
\AgdaDatatype{Transition}%
\>[43]\AgdaSymbol{(}\AgdaInductiveConstructor{inp}\AgdaSpace{}%
\AgdaBound{f}\AgdaSymbol{)}%
\>[52]\AgdaSymbol{(}\AgdaInductiveConstructor{I}\AgdaSpace{}%
\AgdaBound{x}\AgdaSymbol{)}%
\>[59]\AgdaSymbol{(}\AgdaBound{f}\AgdaSpace{}%
\AgdaBound{x}\AgdaSpace{}%
\AgdaSymbol{.}\AgdaField{force}\AgdaSymbol{)}\<%
\\
\>[2]\AgdaInductiveConstructor{out}%
\>[7]\AgdaSymbol{:}\AgdaSpace{}%
\AgdaSymbol{∀\{}\AgdaBound{f}\AgdaSpace{}%
\AgdaBound{x}\AgdaSymbol{\}}\AgdaSpace{}%
\AgdaSymbol{→}\AgdaSpace{}%
\AgdaBound{x}\AgdaSpace{}%
\AgdaOperator{\AgdaFunction{∈}}\AgdaSpace{}%
\AgdaFunction{dom}\AgdaSpace{}%
\AgdaBound{f}%
\>[29]\AgdaSymbol{→}\AgdaSpace{}%
\AgdaDatatype{Transition}%
\>[43]\AgdaSymbol{(}\AgdaInductiveConstructor{out}\AgdaSpace{}%
\AgdaBound{f}\AgdaSymbol{)}%
\>[52]\AgdaSymbol{(}\AgdaInductiveConstructor{O}\AgdaSpace{}%
\AgdaBound{x}\AgdaSymbol{)}%
\>[59]\AgdaSymbol{(}\AgdaBound{f}\AgdaSpace{}%
\AgdaBound{x}\AgdaSpace{}%
\AgdaSymbol{.}\AgdaField{force}\AgdaSymbol{)}\<%
\end{code}

\subsection{Sessions}

In order to specify fair compliance, we need a representation of sessions as pairs $\session{R}{S}$ of session types, just like we have done in Section~\ref{sec:compliance}. To this aim, we introduce the $\AgdaDatatype{Session}$ data type as an alias for pairs of session types.
\begin{code}%
\>[0]\AgdaFunction{Session}\AgdaSpace{}%
\AgdaSymbol{:}\AgdaSpace{}%
\AgdaPrimitive{Set}\<%
\\
\>[0]\AgdaFunction{Session}\AgdaSpace{}%
\AgdaSymbol{=}\AgdaSpace{}%
\AgdaDatatype{SessionType}\AgdaSpace{}%
\AgdaOperator{\AgdaFunction{×}}\AgdaSpace{}%
\AgdaDatatype{SessionType}\<%
\end{code}

A session \emph{reduces} when client and server synchronize, by performing actions with opposite polarities and referring to the same message. We formalize the reduction relation in~\eqref{eq:sred} as the $\AgdaDatatype{Reduction}$ data type, so that a value of type $\AgdaDatatype{Reduction}~(\session{R}{S})~(\session{R'}{S'})$ witnesses the reduction $\session{R}{S} \red \session{R'}{S'}$.
\begin{code}%
\>[0]\AgdaKeyword{data}\AgdaSpace{}%
\AgdaDatatype{Reduction}\AgdaSpace{}%
\AgdaSymbol{:}\AgdaSpace{}%
\AgdaFunction{Session}\AgdaSpace{}%
\AgdaSymbol{→}\AgdaSpace{}%
\AgdaFunction{Session}\AgdaSpace{}%
\AgdaSymbol{→}\AgdaSpace{}%
\AgdaPrimitive{Set}\AgdaSpace{}%
\AgdaKeyword{where}\<%
\\
\>[0][@{}l@{\AgdaIndent{0}}]%
\>[2]\AgdaInductiveConstructor{sync}\AgdaSpace{}%
\AgdaSymbol{:}%
\>[10]\AgdaSymbol{∀\{}\AgdaBound{α}\AgdaSpace{}%
\AgdaBound{R}\AgdaSpace{}%
\AgdaBound{R'}\AgdaSpace{}%
\AgdaBound{S}\AgdaSpace{}%
\AgdaBound{S'}\AgdaSymbol{\}}\AgdaSpace{}%
\AgdaSymbol{→}\AgdaSpace{}%
\AgdaDatatype{Transition}\AgdaSpace{}%
\AgdaBound{R}\AgdaSpace{}%
\AgdaSymbol{(}\AgdaFunction{co-action}\AgdaSpace{}%
\AgdaBound{α}\AgdaSymbol{)}\AgdaSpace{}%
\AgdaBound{R'}\AgdaSpace{}%
\AgdaSymbol{→}\AgdaSpace{}%
\AgdaDatatype{Transition}\AgdaSpace{}%
\AgdaBound{S}\AgdaSpace{}%
\AgdaBound{α}\AgdaSpace{}%
\AgdaBound{S'}\AgdaSpace{}%
\AgdaSymbol{→}\<%
\\
\>[10]\AgdaDatatype{Reduction}\AgdaSpace{}%
\AgdaSymbol{(}\AgdaBound{R}\AgdaSpace{}%
\AgdaOperator{\AgdaInductiveConstructor{,}}\AgdaSpace{}%
\AgdaBound{S}\AgdaSymbol{)}\AgdaSpace{}%
\AgdaSymbol{(}\AgdaBound{R'}\AgdaSpace{}%
\AgdaOperator{\AgdaInductiveConstructor{,}}\AgdaSpace{}%
\AgdaBound{S'}\AgdaSymbol{)}\<%
\end{code}

The weak reduction relation is called $\AgdaDatatype{Reductions}$ and is defined as the reflexive, transitive closure of $\AgdaDatatype{Reduction}$, just like $\wred$ is the reflexive, transitive closure of $\red$. We make use of the $\AgdaDatatype{Star}$ data type from Agda's standard library to define such closure.
\begin{code}%
\>[0]\AgdaFunction{Reductions}\AgdaSpace{}%
\AgdaSymbol{:}\AgdaSpace{}%
\AgdaFunction{Session}\AgdaSpace{}%
\AgdaSymbol{→}\AgdaSpace{}%
\AgdaFunction{Session}\AgdaSpace{}%
\AgdaSymbol{→}\AgdaSpace{}%
\AgdaPrimitive{Set}\<%
\\
\>[0]\AgdaFunction{Reductions}\AgdaSpace{}%
\AgdaSymbol{=}\AgdaSpace{}%
\AgdaDatatype{Star}\AgdaSpace{}%
\AgdaDatatype{Reduction}\<%
\end{code}

\subsection{Fair compliance}

To formalize fair compliance, we define a $\AgdaDatatype{Success}$ predicate that characterizes those configurations $\session{R}{S}$ in which the client has succeeded ($R = \Win$) and the server has not failed ($S \neq \TNil$).
\begin{code}%
\>[0]\AgdaKeyword{data}\AgdaSpace{}%
\AgdaDatatype{Success}\AgdaSpace{}%
\AgdaSymbol{:}\AgdaSpace{}%
\AgdaFunction{Session}\AgdaSpace{}%
\AgdaSymbol{→}\AgdaSpace{}%
\AgdaPrimitive{Set}\AgdaSpace{}%
\AgdaKeyword{where}\<%
\\
\>[0][@{}l@{\AgdaIndent{0}}]%
\>[2]\AgdaInductiveConstructor{success}\AgdaSpace{}%
\AgdaSymbol{:}\AgdaSpace{}%
\AgdaSymbol{∀\{}\AgdaBound{R}\AgdaSpace{}%
\AgdaBound{S}\AgdaSymbol{\}}\AgdaSpace{}%
\AgdaSymbol{→}\AgdaSpace{}%
\AgdaDatatype{Win}\AgdaSpace{}%
\AgdaBound{R}\AgdaSpace{}%
\AgdaSymbol{→}\AgdaSpace{}%
\AgdaDatatype{Defined}\AgdaSpace{}%
\AgdaBound{S}\AgdaSpace{}%
\AgdaSymbol{→}\AgdaSpace{}%
\AgdaDatatype{Success}\AgdaSpace{}%
\AgdaSymbol{(}\AgdaBound{R}\AgdaSpace{}%
\AgdaOperator{\AgdaInductiveConstructor{,}}\AgdaSpace{}%
\AgdaBound{S}\AgdaSymbol{)}\<%
\end{code}

We can then weaken $\AgdaDatatype{Success}$ to $\AgdaDatatype{MaySucceed}$, to characterize those configurations \emph{that can be extended} so as to become successful ones. For this purpose we make use of the $\AgdaDatatype{Satisfiable}$ predicate and of the intersection $\AgdaFunction{∩}$ of two sets from Agda's standard library.
\begin{code}%
\>[0]\AgdaFunction{MaySucceed}\AgdaSpace{}%
\AgdaSymbol{:}\AgdaSpace{}%
\AgdaFunction{Session}\AgdaSpace{}%
\AgdaSymbol{→}\AgdaSpace{}%
\AgdaPrimitive{Set}\<%
\\
\>[0]\AgdaFunction{MaySucceed}\AgdaSpace{}%
\AgdaBound{Se}\AgdaSpace{}%
\AgdaSymbol{=}\AgdaSpace{}%
\AgdaFunction{Relation.Unary.Satisfiable}\AgdaSpace{}%
\AgdaSymbol{(}\AgdaFunction{Reductions}\AgdaSpace{}%
\AgdaBound{Se}\AgdaSpace{}%
\AgdaOperator{\AgdaFunction{∩}}\AgdaSpace{}%
\AgdaDatatype{Success}\AgdaSymbol{)}\<%
\end{code}

In words, $\mathit{Se}$ may succeed if there exists $\mathit{Se}'$ such that $\mathit{Se} \wred \mathit{Se}'$ and $\mathit{Se}'$ is a successful configuration. We can now formulate fair compliance as the property of those sessions that may succeed no matter how they reduce (Definition~\ref{def:fcomp}). This is the specification against which we prove soundness and completeness of the \GIS for fair compliance (Table~\ref{tab:compliance}).
\begin{code}%
\>[0]\AgdaFunction{FCompS}\AgdaSpace{}%
\AgdaSymbol{:}\AgdaSpace{}%
\AgdaFunction{Session}\AgdaSpace{}%
\AgdaSymbol{→}\AgdaSpace{}%
\AgdaPrimitive{Set}\<%
\\
\>[0]\AgdaFunction{FCompS}\AgdaSpace{}%
\AgdaBound{Se}\AgdaSpace{}%
\AgdaSymbol{=}\AgdaSpace{}%
\AgdaSymbol{∀\{}\AgdaBound{Se'}\AgdaSymbol{\}}\AgdaSpace{}%
\AgdaSymbol{→}\AgdaSpace{}%
\AgdaFunction{Reductions}\AgdaSpace{}%
\AgdaBound{Se}\AgdaSpace{}%
\AgdaBound{Se'}\AgdaSpace{}%
\AgdaSymbol{→}\AgdaSpace{}%
\AgdaFunction{MaySucceed}\AgdaSpace{}%
\AgdaBound{Se'}\<%
\end{code}

\subsection{\GIS for fair compliance}

We now use the Agda library for \GISs~\cite{CicconeDagninoZucca21,CicconeDagninoZucca21B} to formally define the inference system for fair compliance shown in Table~\ref{tab:compliance}. The first thing to do is to define the universe $\AgdaDatatype{U}$ of judgments that we want to derive with the inference system. We can equivalently think of  fair compliance as of a binary relation on session types or as a predicate over sessions. We take the second point of view, as it allows us to write more compact code later on.
\begin{code}%
\>[0]\AgdaFunction{U}\AgdaSpace{}%
\AgdaSymbol{:}\AgdaSpace{}%
\AgdaPrimitive{Set}\<%
\\
\>[0]\AgdaFunction{U}\AgdaSpace{}%
\AgdaSymbol{=}\AgdaSpace{}%
\AgdaFunction{Session}\<%
\end{code}

Next, we define two data types to represent the \emph{unique names} with which we identify the rules and corules of the \GIS.\@ We use the same labels of Table~\ref{tab:compliance} except for the corule~\refrule{c-sync} which we split into \emph{two} symmetric corules to avoid reasoning on opposite polarities.
\begin{code}%
\>[0]\AgdaKeyword{data}\AgdaSpace{}%
\AgdaDatatype{RuleNames}\AgdaSpace{}%
\AgdaSymbol{:}\AgdaSpace{}%
\AgdaPrimitive{Set}\AgdaSpace{}%
\AgdaKeyword{where}\<%
\\
\>[0][@{}l@{\AgdaIndent{0}}]%
\>[2]\AgdaInductiveConstructor{success}\AgdaSpace{}%
\AgdaInductiveConstructor{inp-out}\AgdaSpace{}%
\AgdaInductiveConstructor{out-inp}\AgdaSpace{}%
\AgdaSymbol{:}\AgdaSpace{}%
\AgdaDatatype{RuleNames}\<%
\\
\\[\AgdaEmptyExtraSkip]%
\>[0]\AgdaKeyword{data}\AgdaSpace{}%
\AgdaDatatype{CoRuleNames}\AgdaSpace{}%
\AgdaSymbol{:}\AgdaSpace{}%
\AgdaPrimitive{Set}\AgdaSpace{}%
\AgdaKeyword{where}\<%
\\
\>[0][@{}l@{\AgdaIndent{0}}]%
\>[2]\AgdaInductiveConstructor{inp-out}\AgdaSpace{}%
\AgdaInductiveConstructor{out-inp}\AgdaSpace{}%
\AgdaSymbol{:}\AgdaSpace{}%
\AgdaDatatype{CoRuleNames}\<%
\end{code}

There are two different ways of defining rules and corules, depending on whether these have a finite or a possibly infinite number of premises. Clearly, (co)rules with finitely many premises are just a special case of those with possibly infinite ones, but the \GIS library provides some syntactic sugar to specify (co)rules of the former kind in a slightly easier way. We use a finite rule to specify~\refrule{c-success}:
\begin{code}%
\>[0]\AgdaFunction{success-rule}\AgdaSpace{}%
\AgdaSymbol{:}\AgdaSpace{}%
\AgdaRecord{FinMetaRule}\AgdaSpace{}%
\AgdaFunction{U}\<%
\\
\>[0]\AgdaFunction{success-rule}\AgdaSpace{}%
\AgdaSymbol{.}\AgdaField{Ctx}%
\>[29]\AgdaSymbol{=}\AgdaSpace{}%
\AgdaFunction{Σ[}\AgdaSpace{}%
\AgdaBound{Se}\AgdaSpace{}%
\AgdaFunction{∈}\AgdaSpace{}%
\AgdaFunction{Session}\AgdaSpace{}%
\AgdaFunction{]}\AgdaSpace{}%
\AgdaDatatype{Success}\AgdaSpace{}%
\AgdaBound{Se}\<%
\\
\>[0]\AgdaFunction{success-rule}\AgdaSpace{}%
\AgdaSymbol{.}\AgdaField{comp}\AgdaSpace{}%
\AgdaSymbol{(}\AgdaBound{Se}\AgdaSpace{}%
\AgdaOperator{\AgdaInductiveConstructor{,}}\AgdaSpace{}%
\AgdaSymbol{\AgdaUnderscore{})}%
\>[29]\AgdaSymbol{=}\AgdaSpace{}%
\AgdaInductiveConstructor{[]}\AgdaSpace{}%
\AgdaOperator{\AgdaInductiveConstructor{,}}\AgdaSpace{}%
\AgdaBound{Se}\<%
\end{code}

Basically, a finite rule consists of a \emph{context}, a finite list of
\emph{premises} and a \emph{conclusion}. The definition of
$\AgdaFunction{success-rule}$ uses copattern matching to specify an Agda record whose
fields contain such elements. The context field $\AgdaField{Ctx}$ specifies the
type of metavariables occurring in the rule, as well as possible side conditions
that these metavariables are supposed to satisfy. In the specific case of this
rule, we have a single metavariable $\AgdaBound{Se}$ that represents a
successful configuration. The $\AgdaField{comp}$ field is a pair with the list
of premises and the conclusion of the rule. In this case the list is
empty and the conclusion is simply $\AgdaBound{Se}$. Since this field
depends on the metavariable $\AgdaBound{Se}$, we pattern match on a pair to
refer to those components of the context that are needed to express premises and
conclusion.

Concerning~\refrule{c-out-inp} and~\refrule{c-inp-out}, these rules have a possibly infinite set of premises if $\Message$ is infinite. Therefore, we specify the rules using the most general form allowed by the \GIS Agda library. 
\begin{code}%
\>[0]\AgdaFunction{out-inp-rule}\AgdaSpace{}%
\AgdaSymbol{:}\AgdaSpace{}%
\AgdaRecord{MetaRule}\AgdaSpace{}%
\AgdaFunction{U}\<%
\\
\>[0]\AgdaFunction{out-inp-rule}\AgdaSpace{}%
\AgdaSymbol{.}\AgdaField{Ctx}\AgdaSpace{}%
\AgdaSymbol{=}\AgdaSpace{}%
\AgdaFunction{Σ[}\AgdaSpace{}%
\AgdaBound{(}\AgdaBound{f}\AgdaSpace{}%
\AgdaOperator{\AgdaInductiveConstructor{,}}\AgdaSpace{}%
\AgdaBound{\AgdaUnderscore{})}\AgdaSpace{}%
\AgdaFunction{∈}\AgdaSpace{}%
\AgdaFunction{Continuation}\AgdaSpace{}%
\AgdaOperator{\AgdaFunction{×}}\AgdaSpace{}%
\AgdaFunction{Continuation}\AgdaSpace{}%
\AgdaFunction{]}\AgdaSpace{}%
\AgdaFunction{Witness}\AgdaSpace{}%
\AgdaBound{f}\<%
\\
\>[0]\AgdaFunction{out-inp-rule}\AgdaSpace{}%
\AgdaSymbol{.}\AgdaField{Pos}%
\>[22]\AgdaSymbol{((}\AgdaBound{f}\AgdaSpace{}%
\AgdaOperator{\AgdaInductiveConstructor{,}}\AgdaSpace{}%
\AgdaSymbol{\AgdaUnderscore{})}\AgdaSpace{}%
\AgdaOperator{\AgdaInductiveConstructor{,}}\AgdaSpace{}%
\AgdaSymbol{\AgdaUnderscore{})}%
\>[37]\AgdaSymbol{=}\AgdaSpace{}%
\AgdaFunction{Σ[}\AgdaSpace{}%
\AgdaBound{x}\AgdaSpace{}%
\AgdaFunction{∈}\AgdaSpace{}%
\AgdaDatatype{𝕍}\AgdaSpace{}%
\AgdaFunction{]}\AgdaSpace{}%
\AgdaBound{x}\AgdaSpace{}%
\AgdaOperator{\AgdaFunction{∈}}\AgdaSpace{}%
\AgdaFunction{dom}\AgdaSpace{}%
\AgdaBound{f}\<%
\\
\>[0]\AgdaFunction{out-inp-rule}\AgdaSpace{}%
\AgdaSymbol{.}\AgdaField{prems}%
\>[22]\AgdaSymbol{((}\AgdaBound{f}\AgdaSpace{}%
\AgdaOperator{\AgdaInductiveConstructor{,}}\AgdaSpace{}%
\AgdaBound{g}\AgdaSymbol{)}\AgdaSpace{}%
\AgdaOperator{\AgdaInductiveConstructor{,}}\AgdaSpace{}%
\AgdaSymbol{\AgdaUnderscore{})}%
\>[37]\AgdaSymbol{=}\AgdaSpace{}%
\AgdaSymbol{λ}\AgdaSpace{}%
\AgdaSymbol{(}\AgdaBound{x}\AgdaSpace{}%
\AgdaOperator{\AgdaInductiveConstructor{,}}\AgdaSpace{}%
\AgdaSymbol{\AgdaUnderscore{})}\AgdaSpace{}%
\AgdaSymbol{→}\AgdaSpace{}%
\AgdaBound{f}\AgdaSpace{}%
\AgdaBound{x}\AgdaSpace{}%
\AgdaSymbol{.}\AgdaField{force}\AgdaSpace{}%
\AgdaOperator{\AgdaInductiveConstructor{,}}\AgdaSpace{}%
\AgdaBound{g}\AgdaSpace{}%
\AgdaBound{x}\AgdaSpace{}%
\AgdaSymbol{.}\AgdaField{force}\<%
\\
\>[0]\AgdaFunction{out-inp-rule}\AgdaSpace{}%
\AgdaSymbol{.}\AgdaField{conclu}%
\>[22]\AgdaSymbol{((}\AgdaBound{f}\AgdaSpace{}%
\AgdaOperator{\AgdaInductiveConstructor{,}}\AgdaSpace{}%
\AgdaBound{g}\AgdaSymbol{)}\AgdaSpace{}%
\AgdaOperator{\AgdaInductiveConstructor{,}}\AgdaSpace{}%
\AgdaSymbol{\AgdaUnderscore{})}%
\>[37]\AgdaSymbol{=}\AgdaSpace{}%
\AgdaInductiveConstructor{out}\AgdaSpace{}%
\AgdaBound{f}\AgdaSpace{}%
\AgdaOperator{\AgdaInductiveConstructor{,}}\AgdaSpace{}%
\AgdaInductiveConstructor{inp}\AgdaSpace{}%
\AgdaBound{g}\<%
\\
\\[\AgdaEmptyExtraSkip]%
\>[0]\AgdaFunction{inp-out-rule}\AgdaSpace{}%
\AgdaSymbol{:}\AgdaSpace{}%
\AgdaRecord{MetaRule}\AgdaSpace{}%
\AgdaFunction{U}\<%
\\
\>[0]\AgdaFunction{inp-out-rule}\AgdaSpace{}%
\AgdaSymbol{.}\AgdaField{Ctx}\AgdaSpace{}%
\AgdaSymbol{=}\AgdaSpace{}%
\AgdaFunction{Σ[}\AgdaSpace{}%
\AgdaBound{(\AgdaUnderscore{}}\AgdaSpace{}%
\AgdaOperator{\AgdaInductiveConstructor{,}}\AgdaSpace{}%
\AgdaBound{g}\AgdaBound{)}\AgdaSpace{}%
\AgdaFunction{∈}\AgdaSpace{}%
\AgdaFunction{Continuation}\AgdaSpace{}%
\AgdaOperator{\AgdaFunction{×}}\AgdaSpace{}%
\AgdaFunction{Continuation}\AgdaSpace{}%
\AgdaFunction{]}\AgdaSpace{}%
\AgdaFunction{Witness}\AgdaSpace{}%
\AgdaBound{g}\<%
\\
\>[0]\AgdaFunction{inp-out-rule}\AgdaSpace{}%
\AgdaSymbol{.}\AgdaField{Pos}%
\>[22]\AgdaSymbol{((\AgdaUnderscore{}}\AgdaSpace{}%
\AgdaOperator{\AgdaInductiveConstructor{,}}\AgdaSpace{}%
\AgdaBound{g}\AgdaSymbol{)}\AgdaSpace{}%
\AgdaOperator{\AgdaInductiveConstructor{,}}\AgdaSpace{}%
\AgdaSymbol{\AgdaUnderscore{})}%
\>[37]\AgdaSymbol{=}\AgdaSpace{}%
\AgdaFunction{Σ[}\AgdaSpace{}%
\AgdaBound{x}\AgdaSpace{}%
\AgdaFunction{∈}\AgdaSpace{}%
\AgdaDatatype{𝕍}\AgdaSpace{}%
\AgdaFunction{]}\AgdaSpace{}%
\AgdaBound{x}\AgdaSpace{}%
\AgdaOperator{\AgdaFunction{∈}}\AgdaSpace{}%
\AgdaFunction{dom}\AgdaSpace{}%
\AgdaBound{g}\<%
\\
\>[0]\AgdaFunction{inp-out-rule}\AgdaSpace{}%
\AgdaSymbol{.}\AgdaField{prems}%
\>[22]\AgdaSymbol{((}\AgdaBound{f}\AgdaSpace{}%
\AgdaOperator{\AgdaInductiveConstructor{,}}\AgdaSpace{}%
\AgdaBound{g}\AgdaSymbol{)}\AgdaSpace{}%
\AgdaOperator{\AgdaInductiveConstructor{,}}\AgdaSpace{}%
\AgdaSymbol{\AgdaUnderscore{})}%
\>[37]\AgdaSymbol{=}\AgdaSpace{}%
\AgdaSymbol{λ}\AgdaSpace{}%
\AgdaSymbol{(}\AgdaBound{x}\AgdaSpace{}%
\AgdaOperator{\AgdaInductiveConstructor{,}}\AgdaSpace{}%
\AgdaSymbol{\AgdaUnderscore{})}\AgdaSpace{}%
\AgdaSymbol{→}\AgdaSpace{}%
\AgdaBound{f}\AgdaSpace{}%
\AgdaBound{x}\AgdaSpace{}%
\AgdaSymbol{.}\AgdaField{force}\AgdaSpace{}%
\AgdaOperator{\AgdaInductiveConstructor{,}}\AgdaSpace{}%
\AgdaBound{g}\AgdaSpace{}%
\AgdaBound{x}\AgdaSpace{}%
\AgdaSymbol{.}\AgdaField{force}\<%
\\
\>[0]\AgdaFunction{inp-out-rule}\AgdaSpace{}%
\AgdaSymbol{.}\AgdaField{conclu}%
\>[22]\AgdaSymbol{((}\AgdaBound{f}\AgdaSpace{}%
\AgdaOperator{\AgdaInductiveConstructor{,}}\AgdaSpace{}%
\AgdaBound{g}\AgdaSymbol{)}\AgdaSpace{}%
\AgdaOperator{\AgdaInductiveConstructor{,}}\AgdaSpace{}%
\AgdaSymbol{\AgdaUnderscore{})}%
\>[37]\AgdaSymbol{=}\AgdaSpace{}%
\AgdaInductiveConstructor{inp}\AgdaSpace{}%
\AgdaBound{f}\AgdaSpace{}%
\AgdaOperator{\AgdaInductiveConstructor{,}}\AgdaSpace{}%
\AgdaInductiveConstructor{out}\AgdaSpace{}%
\AgdaBound{g}\<%
\end{code}

Again, a rule with possibly infinite premises is specified as an Agda
record, this time with four fields: the $\AgdaField{Ctx}$ field provides the
type of the metavariables along with possible side conditions, just as in the
case of finite rules. All the subsequent fields use pattern matching
to access the needed parts of the context. The $\AgdaField{Pos}$ field is the
\emph{domain} of the function that \emph{generates} the premises given a
context. In the above rules, the domain coincides with that of the continuation
function corresponding to the output session type, since we want to specify a
fair compliance premise for every message that can be sent. We can think of
$\AgdaField{Pos}$ as of the type of the \emph{position} of each premise above
the line of a rule. The $\AgdaField{prems}$ field contains the function that,
given a context and a position, yields the premise found at that position.
Above, the premise is the pair of continuations after an exchange of a message
$\AgdaBound{x}$. Finally, the $\AgdaField{conclu}$ field contains the conclusion
of the rule.

The specification of corules is no different from that of plain rules. As we have anticipated, we split~\refrule{c-sync} into two corules, each having exactly one premise.
\begin{code}%
\>[0]\AgdaFunction{out-inp-corule}\AgdaSpace{}%
\AgdaSymbol{:}\AgdaSpace{}%
\AgdaRecord{FinMetaRule}\AgdaSpace{}%
\AgdaFunction{U}\<%
\\
\>[0]\AgdaFunction{out-inp-corule}\AgdaSpace{}%
\AgdaSymbol{.}\AgdaField{Ctx}\AgdaSpace{}%
\AgdaSymbol{=}\AgdaSpace{}%
\AgdaFunction{Σ[}\AgdaSpace{}%
\AgdaBound{(}\AgdaBound{f}\AgdaSpace{}%
\AgdaOperator{\AgdaInductiveConstructor{,}}\AgdaSpace{}%
\AgdaBound{\AgdaUnderscore{})}\AgdaSpace{}%
\AgdaFunction{∈}\AgdaSpace{}%
\AgdaFunction{Continuation}\AgdaSpace{}%
\AgdaOperator{\AgdaFunction{×}}\AgdaSpace{}%
\AgdaFunction{Continuation}\AgdaSpace{}%
\AgdaFunction{]}\AgdaSpace{}%
\AgdaFunction{Witness}\AgdaSpace{}%
\AgdaBound{f}\<%
\\
\>[0]\AgdaFunction{out-inp-corule}\AgdaSpace{}%
\AgdaSymbol{.}\AgdaField{comp}\AgdaSpace{}%
\AgdaSymbol{((}\AgdaBound{f}\AgdaSpace{}%
\AgdaOperator{\AgdaInductiveConstructor{,}}\AgdaSpace{}%
\AgdaBound{g}\AgdaSymbol{)}\AgdaSpace{}%
\AgdaOperator{\AgdaInductiveConstructor{,}}\AgdaSpace{}%
\AgdaBound{x}\AgdaSpace{}%
\AgdaOperator{\AgdaInductiveConstructor{,}}\AgdaSpace{}%
\AgdaSymbol{\AgdaUnderscore{})}%
\>[549I]\AgdaSymbol{=}\AgdaSpace{}%
\AgdaSymbol{(}\AgdaBound{f}\AgdaSpace{}%
\AgdaBound{x}\AgdaSpace{}%
\AgdaSymbol{.}\AgdaField{force}\AgdaSpace{}%
\AgdaOperator{\AgdaInductiveConstructor{,}}\AgdaSpace{}%
\AgdaBound{g}\AgdaSpace{}%
\AgdaBound{x}\AgdaSpace{}%
\AgdaSymbol{.}\AgdaField{force}\AgdaSymbol{)}\AgdaSpace{}%
\AgdaOperator{\AgdaInductiveConstructor{∷}}\AgdaSpace{}%
\AgdaInductiveConstructor{[]}\<%
\\
\>[.][@{}l@{}]\<[549I]%
\>[39]\AgdaOperator{\AgdaInductiveConstructor{,}}\AgdaSpace{}%
\AgdaSymbol{(}\AgdaInductiveConstructor{out}\AgdaSpace{}%
\AgdaBound{f}\AgdaSpace{}%
\AgdaOperator{\AgdaInductiveConstructor{,}}\AgdaSpace{}%
\AgdaInductiveConstructor{inp}\AgdaSpace{}%
\AgdaBound{g}\AgdaSymbol{)}\<%
\\
\\[\AgdaEmptyExtraSkip]%
\>[0]\AgdaFunction{inp-out-corule}\AgdaSpace{}%
\AgdaSymbol{:}\AgdaSpace{}%
\AgdaRecord{FinMetaRule}\AgdaSpace{}%
\AgdaFunction{U}\<%
\\
\>[0]\AgdaFunction{inp-out-corule}\AgdaSpace{}%
\AgdaSymbol{.}\AgdaField{Ctx}\AgdaSpace{}%
\AgdaSymbol{=}\AgdaSpace{}%
\AgdaFunction{Σ[}\AgdaSpace{}%
\AgdaBound{(\AgdaUnderscore{}}\AgdaSpace{}%
\AgdaOperator{\AgdaInductiveConstructor{,}}\AgdaSpace{}%
\AgdaBound{g}\AgdaBound{)}\AgdaSpace{}%
\AgdaFunction{∈}\AgdaSpace{}%
\AgdaFunction{Continuation}\AgdaSpace{}%
\AgdaOperator{\AgdaFunction{×}}\AgdaSpace{}%
\AgdaFunction{Continuation}\AgdaSpace{}%
\AgdaFunction{]}\AgdaSpace{}%
\AgdaFunction{Witness}\AgdaSpace{}%
\AgdaBound{g}\<%
\\
\>[0]\AgdaFunction{inp-out-corule}\AgdaSpace{}%
\AgdaSymbol{.}\AgdaField{comp}\AgdaSpace{}%
\AgdaSymbol{((}\AgdaBound{f}\AgdaSpace{}%
\AgdaOperator{\AgdaInductiveConstructor{,}}\AgdaSpace{}%
\AgdaBound{g}\AgdaSymbol{)}\AgdaSpace{}%
\AgdaOperator{\AgdaInductiveConstructor{,}}\AgdaSpace{}%
\AgdaBound{x}\AgdaSpace{}%
\AgdaOperator{\AgdaInductiveConstructor{,}}\AgdaSpace{}%
\AgdaSymbol{\AgdaUnderscore{})}%
\>[588I]\AgdaSymbol{=}\AgdaSpace{}%
\AgdaSymbol{(}\AgdaBound{f}\AgdaSpace{}%
\AgdaBound{x}\AgdaSpace{}%
\AgdaSymbol{.}\AgdaField{force}\AgdaSpace{}%
\AgdaOperator{\AgdaInductiveConstructor{,}}\AgdaSpace{}%
\AgdaBound{g}\AgdaSpace{}%
\AgdaBound{x}\AgdaSpace{}%
\AgdaSymbol{.}\AgdaField{force}\AgdaSymbol{)}\AgdaSpace{}%
\AgdaOperator{\AgdaInductiveConstructor{∷}}\AgdaSpace{}%
\AgdaInductiveConstructor{[]}\<%
\\
\>[.][@{}l@{}]\<[588I]%
\>[39]\AgdaOperator{\AgdaInductiveConstructor{,}}\AgdaSpace{}%
\AgdaSymbol{(}\AgdaInductiveConstructor{inp}\AgdaSpace{}%
\AgdaBound{f}\AgdaSpace{}%
\AgdaOperator{\AgdaInductiveConstructor{,}}\AgdaSpace{}%
\AgdaInductiveConstructor{out}\AgdaSpace{}%
\AgdaBound{g}\AgdaSymbol{)}\<%
\end{code}

We can now define two inference systems, $\AgdaFunction{FCompIS}$ that consists of the plain rules only and $\AgdaFunction{FCompCOIS}$ that consists of the corules only. These are called $\is[C]$ and $\cois[C]$ in Section~\ref{sec:compliance}.
\begin{code}%
\>[0]\AgdaFunction{FCompIS}\AgdaSpace{}%
\AgdaSymbol{:}\AgdaSpace{}%
\AgdaRecord{IS}\AgdaSpace{}%
\AgdaFunction{U}\<%
\\
\>[0]\AgdaFunction{FCompIS}\AgdaSpace{}%
\AgdaSymbol{.}\AgdaField{Names}%
\>[24]\AgdaSymbol{=}\AgdaSpace{}%
\AgdaDatatype{RuleNames}\<%
\\
\>[0]\AgdaFunction{FCompIS}\AgdaSpace{}%
\AgdaSymbol{.}\AgdaField{rules}\AgdaSpace{}%
\AgdaInductiveConstructor{success}%
\>[24]\AgdaSymbol{=}\AgdaSpace{}%
\AgdaFunction{from}\AgdaSpace{}%
\AgdaFunction{success-rule}\<%
\\
\>[0]\AgdaFunction{FCompIS}\AgdaSpace{}%
\AgdaSymbol{.}\AgdaField{rules}\AgdaSpace{}%
\AgdaInductiveConstructor{out-inp}%
\>[24]\AgdaSymbol{=}\AgdaSpace{}%
\AgdaFunction{out-inp-rule}\<%
\\
\>[0]\AgdaFunction{FCompIS}\AgdaSpace{}%
\AgdaSymbol{.}\AgdaField{rules}\AgdaSpace{}%
\AgdaInductiveConstructor{inp-out}%
\>[24]\AgdaSymbol{=}\AgdaSpace{}%
\AgdaFunction{inp-out-rule}\<%
\\
\\[\AgdaEmptyExtraSkip]%
\>[0]\AgdaFunction{FCompCOIS}\AgdaSpace{}%
\AgdaSymbol{:}\AgdaSpace{}%
\AgdaRecord{IS}\AgdaSpace{}%
\AgdaFunction{U}\<%
\\
\>[0]\AgdaFunction{FCompCOIS}\AgdaSpace{}%
\AgdaSymbol{.}\AgdaField{Names}%
\>[26]\AgdaSymbol{=}\AgdaSpace{}%
\AgdaDatatype{CoRuleNames}\<%
\\
\>[0]\AgdaFunction{FCompCOIS}\AgdaSpace{}%
\AgdaSymbol{.}\AgdaField{rules}\AgdaSpace{}%
\AgdaInductiveConstructor{out-inp}%
\>[26]\AgdaSymbol{=}\AgdaSpace{}%
\AgdaFunction{from}\AgdaSpace{}%
\AgdaFunction{out-inp-corule}\<%
\\
\>[0]\AgdaFunction{FCompCOIS}\AgdaSpace{}%
\AgdaSymbol{.}\AgdaField{rules}\AgdaSpace{}%
\AgdaInductiveConstructor{inp-out}%
\>[26]\AgdaSymbol{=}\AgdaSpace{}%
\AgdaFunction{from}\AgdaSpace{}%
\AgdaFunction{inp-out-corule}\<%
\end{code}

An inference system is encoded as an Agda record with two fields:
$\AgdaField{Names}$ is the data type representing the names of the rules,
whereas $\AgdaField{rules}$ is a function associating each name to the
corresponding rule. We use copatterns to define the fields of such a record and
pattern matching to discriminate each rule. The auxiliary function
$\AgdaFunction{from}$ converts a finite rule into its more general
form on-the-fly, so that the internal representation of all rules is uniform.

We obtain the generalized interpretation of $\gis{\is[C]}{\cois[C]}$, which we call $\AgdaFunction{FCompG}$, through the library function $\AgdaFunction{Gen}$.
\begin{code}%
\>[0]\AgdaFunction{FCompG}\AgdaSpace{}%
\AgdaSymbol{:}\AgdaSpace{}%
\AgdaFunction{Session}\AgdaSpace{}%
\AgdaSymbol{→}\AgdaSpace{}%
\AgdaPrimitive{Set}\<%
\\
\>[0]\AgdaFunction{FCompG}\AgdaSpace{}%
\AgdaSymbol{=}\AgdaSpace{}%
\AgdaOperator{\AgdaFunction{Gen⟦}}\AgdaSpace{}%
\AgdaFunction{FCompIS}\AgdaSpace{}%
\AgdaOperator{\AgdaFunction{,}}\AgdaSpace{}%
\AgdaFunction{FCompCOIS}\AgdaSpace{}%
\AgdaOperator{\AgdaFunction{⟧}}\<%
\end{code}

The relation $\compliance{}{}$ defined by the \GIS in Table~\ref{tab:compliance} is now just a curried version of $\AgdaFunction{FCompG}$.
\begin{code}%
\>[0]\AgdaOperator{\AgdaFunction{\AgdaUnderscore{}⊣\AgdaUnderscore{}}}\AgdaSpace{}%
\AgdaSymbol{:}\AgdaSpace{}%
\AgdaDatatype{SessionType}\AgdaSpace{}%
\AgdaSymbol{→}\AgdaSpace{}%
\AgdaDatatype{SessionType}\AgdaSpace{}%
\AgdaSymbol{→}\AgdaSpace{}%
\AgdaPrimitive{Set}\<%
\\
\>[0]\AgdaBound{R}\AgdaSpace{}%
\AgdaOperator{\AgdaFunction{⊣}}\AgdaSpace{}%
\AgdaBound{S}\AgdaSpace{}%
\AgdaSymbol{=}\AgdaSpace{}%
\AgdaFunction{FCompG}\AgdaSpace{}%
\AgdaSymbol{(}\AgdaBound{R}\AgdaSpace{}%
\AgdaOperator{\AgdaInductiveConstructor{,}}\AgdaSpace{}%
\AgdaBound{S}\AgdaSymbol{)}\<%
\end{code}

We also define a predicate $\AgdaFunction{FCompI}$ as the inductive interpretation of the union of $\AgdaFunction{FCompIS}$ and $\AgdaFunction{FCompCOIS}$, which is useful in the soundness and boundedness proofs of the \GIS.
\begin{code}%
\>[0]\AgdaFunction{FCompI}\AgdaSpace{}%
\AgdaSymbol{:}\AgdaSpace{}%
\AgdaFunction{Session}\AgdaSpace{}%
\AgdaSymbol{→}\AgdaSpace{}%
\AgdaPrimitive{Set}\<%
\\
\>[0]\AgdaFunction{FCompI}\AgdaSpace{}%
\AgdaSymbol{=}\AgdaSpace{}%
\AgdaOperator{\AgdaDatatype{Ind⟦}}\AgdaSpace{}%
\AgdaFunction{FCompIS}\AgdaSpace{}%
\AgdaOperator{\AgdaFunction{∪}}\AgdaSpace{}%
\AgdaFunction{FCompCOIS}\AgdaSpace{}%
\AgdaOperator{\AgdaDatatype{⟧}}\<%
\end{code}

\subsection{Soundness}

\GISs provide no canonical way for proving the soundness of the generalized interpretation of an inference system, so we have to handcraft the proof. We start by proving that the inductive interpretation of the inference system with the corules implies the existence of a reduction leading to a successful configuration.
\begin{code}%
\>[0]\AgdaFunction{FCompI→MaySucceed}\AgdaSpace{}%
\AgdaSymbol{:}\AgdaSpace{}%
\AgdaSymbol{∀\{}\AgdaBound{Se}\AgdaSymbol{\}}\AgdaSpace{}%
\AgdaSymbol{→}\AgdaSpace{}%
\AgdaFunction{FCompI}\AgdaSpace{}%
\AgdaBound{Se}\AgdaSpace{}%
\AgdaSymbol{→}\AgdaSpace{}%
\AgdaFunction{MaySucceed}\AgdaSpace{}%
\AgdaBound{Se}\<%
\\
\>[0]\AgdaFunction{FCompI→MaySucceed}\AgdaSpace{}%
\AgdaSymbol{(}\AgdaInductiveConstructor{fold}\AgdaSpace{}%
\AgdaSymbol{(}\AgdaInductiveConstructor{inj₁}\AgdaSpace{}%
\AgdaInductiveConstructor{success}\AgdaSpace{}%
\AgdaOperator{\AgdaInductiveConstructor{,}}\AgdaSpace{}%
\AgdaSymbol{(\AgdaUnderscore{}}\AgdaSpace{}%
\AgdaOperator{\AgdaInductiveConstructor{,}}\AgdaSpace{}%
\AgdaBound{succ}\AgdaSymbol{)}\AgdaSpace{}%
\AgdaOperator{\AgdaInductiveConstructor{,}}\AgdaSpace{}%
\AgdaInductiveConstructor{refl}\AgdaSpace{}%
\AgdaOperator{\AgdaInductiveConstructor{,}}\AgdaSpace{}%
\AgdaSymbol{\AgdaUnderscore{}))}\AgdaSpace{}%
\AgdaSymbol{=}\AgdaSpace{}%
\AgdaSymbol{\AgdaUnderscore{}}\AgdaSpace{}%
\AgdaOperator{\AgdaInductiveConstructor{,}}\AgdaSpace{}%
\AgdaInductiveConstructor{ε}\AgdaSpace{}%
\AgdaOperator{\AgdaInductiveConstructor{,}}\AgdaSpace{}%
\AgdaBound{succ}\<%
\\
\>[0]\AgdaFunction{FCompI→MaySucceed}\AgdaSpace{}%
\AgdaSymbol{(}\AgdaInductiveConstructor{fold}\AgdaSpace{}%
\AgdaSymbol{(}\AgdaInductiveConstructor{inj₁}\AgdaSpace{}%
\AgdaInductiveConstructor{out-inp}\AgdaSpace{}%
\AgdaOperator{\AgdaInductiveConstructor{,}}\AgdaSpace{}%
\AgdaSymbol{(\AgdaUnderscore{}}\AgdaSpace{}%
\AgdaOperator{\AgdaInductiveConstructor{,}}\AgdaSpace{}%
\AgdaSymbol{\AgdaUnderscore{}}\AgdaSpace{}%
\AgdaOperator{\AgdaInductiveConstructor{,}}\AgdaSpace{}%
\AgdaBound{fx}\AgdaSymbol{)}\AgdaSpace{}%
\AgdaOperator{\AgdaInductiveConstructor{,}}\AgdaSpace{}%
\AgdaInductiveConstructor{refl}\AgdaSpace{}%
\AgdaOperator{\AgdaInductiveConstructor{,}}\AgdaSpace{}%
\AgdaBound{pr}\AgdaSymbol{))}\AgdaSpace{}%
\AgdaSymbol{=}\<%
\\
\>[0][@{}l@{\AgdaIndent{0}}]%
\>[2]\AgdaKeyword{let}\AgdaSpace{}%
\AgdaSymbol{\AgdaUnderscore{}}\AgdaSpace{}%
\AgdaOperator{\AgdaInductiveConstructor{,}}\AgdaSpace{}%
\AgdaBound{reds}\AgdaSpace{}%
\AgdaOperator{\AgdaInductiveConstructor{,}}\AgdaSpace{}%
\AgdaBound{succ}\AgdaSpace{}%
\AgdaSymbol{=}\AgdaSpace{}%
\AgdaFunction{FCompI→MaySucceed}\AgdaSpace{}%
\AgdaSymbol{(}\AgdaBound{pr}\AgdaSpace{}%
\AgdaSymbol{(\AgdaUnderscore{}}\AgdaSpace{}%
\AgdaOperator{\AgdaInductiveConstructor{,}}\AgdaSpace{}%
\AgdaBound{fx}\AgdaSymbol{))}\AgdaSpace{}%
\AgdaKeyword{in}\<%
\\
\>[2]\AgdaSymbol{\AgdaUnderscore{}}\AgdaSpace{}%
\AgdaOperator{\AgdaInductiveConstructor{,}}\AgdaSpace{}%
\AgdaInductiveConstructor{sync}\AgdaSpace{}%
\AgdaSymbol{(}\AgdaInductiveConstructor{out}\AgdaSpace{}%
\AgdaBound{fx}\AgdaSymbol{)}\AgdaSpace{}%
\AgdaInductiveConstructor{inp}\AgdaSpace{}%
\AgdaOperator{\AgdaInductiveConstructor{◅}}\AgdaSpace{}%
\AgdaBound{reds}\AgdaSpace{}%
\AgdaOperator{\AgdaInductiveConstructor{,}}\AgdaSpace{}%
\AgdaBound{succ}\<%
\\
\>[0]\AgdaFunction{FCompI→MaySucceed}\AgdaSpace{}%
\AgdaSymbol{(}\AgdaInductiveConstructor{fold}\AgdaSpace{}%
\AgdaSymbol{(}\AgdaInductiveConstructor{inj₁}\AgdaSpace{}%
\AgdaInductiveConstructor{inp-out}\AgdaSpace{}%
\AgdaOperator{\AgdaInductiveConstructor{,}}\AgdaSpace{}%
\AgdaSymbol{(\AgdaUnderscore{}}\AgdaSpace{}%
\AgdaOperator{\AgdaInductiveConstructor{,}}\AgdaSpace{}%
\AgdaSymbol{\AgdaUnderscore{}}\AgdaSpace{}%
\AgdaOperator{\AgdaInductiveConstructor{,}}\AgdaSpace{}%
\AgdaBound{gx}\AgdaSymbol{)}\AgdaSpace{}%
\AgdaOperator{\AgdaInductiveConstructor{,}}\AgdaSpace{}%
\AgdaInductiveConstructor{refl}\AgdaSpace{}%
\AgdaOperator{\AgdaInductiveConstructor{,}}\AgdaSpace{}%
\AgdaBound{pr}\AgdaSymbol{))}\AgdaSpace{}%
\AgdaSymbol{=}\<%
\\
\>[0][@{}l@{\AgdaIndent{0}}]%
\>[2]\AgdaKeyword{let}\AgdaSpace{}%
\AgdaSymbol{\AgdaUnderscore{}}\AgdaSpace{}%
\AgdaOperator{\AgdaInductiveConstructor{,}}\AgdaSpace{}%
\AgdaBound{reds}\AgdaSpace{}%
\AgdaOperator{\AgdaInductiveConstructor{,}}\AgdaSpace{}%
\AgdaBound{succ}\AgdaSpace{}%
\AgdaSymbol{=}\AgdaSpace{}%
\AgdaFunction{FCompI→MaySucceed}\AgdaSpace{}%
\AgdaSymbol{(}\AgdaBound{pr}\AgdaSpace{}%
\AgdaSymbol{(\AgdaUnderscore{}}\AgdaSpace{}%
\AgdaOperator{\AgdaInductiveConstructor{,}}\AgdaSpace{}%
\AgdaBound{gx}\AgdaSymbol{))}\AgdaSpace{}%
\AgdaKeyword{in}\<%
\\
\>[2]\AgdaSymbol{\AgdaUnderscore{}}\AgdaSpace{}%
\AgdaOperator{\AgdaInductiveConstructor{,}}\AgdaSpace{}%
\AgdaInductiveConstructor{sync}\AgdaSpace{}%
\AgdaInductiveConstructor{inp}\AgdaSpace{}%
\AgdaSymbol{(}\AgdaInductiveConstructor{out}\AgdaSpace{}%
\AgdaBound{gx}\AgdaSymbol{)}\AgdaSpace{}%
\AgdaOperator{\AgdaInductiveConstructor{◅}}\AgdaSpace{}%
\AgdaBound{reds}\AgdaSpace{}%
\AgdaOperator{\AgdaInductiveConstructor{,}}\AgdaSpace{}%
\AgdaBound{succ}\<%
\\
\>[0]\AgdaFunction{FCompI→MaySucceed}\AgdaSpace{}%
\AgdaSymbol{(}\AgdaInductiveConstructor{fold}\AgdaSpace{}%
\AgdaSymbol{(}\AgdaInductiveConstructor{inj₂}\AgdaSpace{}%
\AgdaInductiveConstructor{out-inp}\AgdaSpace{}%
\AgdaOperator{\AgdaInductiveConstructor{,}}\AgdaSpace{}%
\AgdaSymbol{(\AgdaUnderscore{}}\AgdaSpace{}%
\AgdaOperator{\AgdaInductiveConstructor{,}}\AgdaSpace{}%
\AgdaSymbol{\AgdaUnderscore{}}\AgdaSpace{}%
\AgdaOperator{\AgdaInductiveConstructor{,}}\AgdaSpace{}%
\AgdaBound{fx}\AgdaSymbol{)}\AgdaSpace{}%
\AgdaOperator{\AgdaInductiveConstructor{,}}\AgdaSpace{}%
\AgdaInductiveConstructor{refl}\AgdaSpace{}%
\AgdaOperator{\AgdaInductiveConstructor{,}}\AgdaSpace{}%
\AgdaBound{pr}\AgdaSymbol{))}\AgdaSpace{}%
\AgdaSymbol{=}\<%
\\
\>[0][@{}l@{\AgdaIndent{0}}]%
\>[2]\AgdaKeyword{let}\AgdaSpace{}%
\AgdaSymbol{\AgdaUnderscore{}}\AgdaSpace{}%
\AgdaOperator{\AgdaInductiveConstructor{,}}\AgdaSpace{}%
\AgdaBound{reds}\AgdaSpace{}%
\AgdaOperator{\AgdaInductiveConstructor{,}}\AgdaSpace{}%
\AgdaBound{succ}\AgdaSpace{}%
\AgdaSymbol{=}\AgdaSpace{}%
\AgdaFunction{FCompI→MaySucceed}\AgdaSpace{}%
\AgdaSymbol{(}\AgdaBound{pr}\AgdaSpace{}%
\AgdaInductiveConstructor{Data.Fin.zero}\AgdaSymbol{)}\AgdaSpace{}%
\AgdaKeyword{in}\<%
\\
\>[2]\AgdaSymbol{\AgdaUnderscore{}}\AgdaSpace{}%
\AgdaOperator{\AgdaInductiveConstructor{,}}\AgdaSpace{}%
\AgdaInductiveConstructor{sync}\AgdaSpace{}%
\AgdaSymbol{(}\AgdaInductiveConstructor{out}\AgdaSpace{}%
\AgdaBound{fx}\AgdaSymbol{)}\AgdaSpace{}%
\AgdaInductiveConstructor{inp}\AgdaSpace{}%
\AgdaOperator{\AgdaInductiveConstructor{◅}}\AgdaSpace{}%
\AgdaBound{reds}\AgdaSpace{}%
\AgdaOperator{\AgdaInductiveConstructor{,}}\AgdaSpace{}%
\AgdaBound{succ}\<%
\\
\>[0]\AgdaFunction{FCompI→MaySucceed}\AgdaSpace{}%
\AgdaSymbol{(}\AgdaInductiveConstructor{fold}\AgdaSpace{}%
\AgdaSymbol{(}\AgdaInductiveConstructor{inj₂}\AgdaSpace{}%
\AgdaInductiveConstructor{inp-out}\AgdaSpace{}%
\AgdaOperator{\AgdaInductiveConstructor{,}}\AgdaSpace{}%
\AgdaSymbol{(\AgdaUnderscore{}}\AgdaSpace{}%
\AgdaOperator{\AgdaInductiveConstructor{,}}\AgdaSpace{}%
\AgdaSymbol{\AgdaUnderscore{}}\AgdaSpace{}%
\AgdaOperator{\AgdaInductiveConstructor{,}}\AgdaSpace{}%
\AgdaBound{gx}\AgdaSymbol{)}\AgdaSpace{}%
\AgdaOperator{\AgdaInductiveConstructor{,}}\AgdaSpace{}%
\AgdaInductiveConstructor{refl}\AgdaSpace{}%
\AgdaOperator{\AgdaInductiveConstructor{,}}\AgdaSpace{}%
\AgdaBound{pr}\AgdaSymbol{))}\AgdaSpace{}%
\AgdaSymbol{=}\<%
\\
\>[0][@{}l@{\AgdaIndent{0}}]%
\>[2]\AgdaKeyword{let}\AgdaSpace{}%
\AgdaSymbol{\AgdaUnderscore{}}\AgdaSpace{}%
\AgdaOperator{\AgdaInductiveConstructor{,}}\AgdaSpace{}%
\AgdaBound{reds}\AgdaSpace{}%
\AgdaOperator{\AgdaInductiveConstructor{,}}\AgdaSpace{}%
\AgdaBound{succ}\AgdaSpace{}%
\AgdaSymbol{=}\AgdaSpace{}%
\AgdaFunction{FCompI→MaySucceed}\AgdaSpace{}%
\AgdaSymbol{(}\AgdaBound{pr}\AgdaSpace{}%
\AgdaInductiveConstructor{Data.Fin.zero}\AgdaSymbol{)}\AgdaSpace{}%
\AgdaKeyword{in}\<%
\\
\>[2]\AgdaSymbol{\AgdaUnderscore{}}\AgdaSpace{}%
\AgdaOperator{\AgdaInductiveConstructor{,}}\AgdaSpace{}%
\AgdaInductiveConstructor{sync}\AgdaSpace{}%
\AgdaInductiveConstructor{inp}\AgdaSpace{}%
\AgdaSymbol{(}\AgdaInductiveConstructor{out}\AgdaSpace{}%
\AgdaBound{gx}\AgdaSymbol{)}\AgdaSpace{}%
\AgdaOperator{\AgdaInductiveConstructor{◅}}\AgdaSpace{}%
\AgdaBound{reds}\AgdaSpace{}%
\AgdaOperator{\AgdaInductiveConstructor{,}}\AgdaSpace{}%
\AgdaBound{succ}\<%
\end{code}

There are two things worth noting here. First, in the union of the two inference systems each rule is identified by a name of the form $\AgdaInductiveConstructor{inj₁}~n$ or $\AgdaInductiveConstructor{inj₂}~n$ where $n$ is either the name of a rule or of a corule, respectively.
Also, we use the function $\AgdaBound{pr}$ to access the premises of a (co)rule by their position. For the plain rules $\AgdaInductiveConstructor{out-inp}$ and $\AgdaInductiveConstructor{inp-out}$ the position is the witness $\AgdaBound{fx}$ or $\AgdaBound{gx}$ that the value exchanged in the synchronization belongs to the domain of the continuation function of the sender. For the corules $\AgdaInductiveConstructor{out-inp}$ and $\AgdaInductiveConstructor{inp-out}$, we use the position $\AgdaInductiveConstructor{Data.Fin.zero}$ to access the first and only premise in their list of premises.

The next auxiliary result establishes a ``subject reduction'' property for fair compliance: if $\compliance{R}{S}$ and $\session{R}{S} \wred \session{R'}{S'}$, then $\compliance{R'}{S'}$. Note that this property is trivial to prove when we consider the specification of fair compliance (Definition~\ref{def:fcomp}), but here we are referring to the predicate defined by the \GIS.\@ The proof consists of a simple induction on the reduction $\session{R}{S} \wred \session{R'}{S'}$, where $\AgdaInductiveConstructor{ε}$ (constructor of $\AgdaFunction{Star}$) represents the base case (when there are no reductions) and $\AgdaBound{red}~\AgdaInductiveConstructor{◅}~\AgdaBound{reds}$ represents a chain of reductions starting with the single reduction $\AgdaBound{red}$ followed by the reductions $\AgdaBound{reds}$.
\begin{code}%
\>[0]\AgdaFunction{sr}\AgdaSpace{}%
\AgdaSymbol{:}\AgdaSpace{}%
\AgdaSymbol{∀\{}\AgdaBound{Se}\AgdaSpace{}%
\AgdaBound{Se'}\AgdaSymbol{\}}\AgdaSpace{}%
\AgdaSymbol{→}\AgdaSpace{}%
\AgdaFunction{FCompG}\AgdaSpace{}%
\AgdaBound{Se}\AgdaSpace{}%
\AgdaSymbol{→}\AgdaSpace{}%
\AgdaFunction{Reductions}\AgdaSpace{}%
\AgdaBound{Se}\AgdaSpace{}%
\AgdaBound{Se'}\AgdaSpace{}%
\AgdaSymbol{→}\AgdaSpace{}%
\AgdaFunction{FCompG}\AgdaSpace{}%
\AgdaBound{Se'}\<%
\\
\>[0]\AgdaFunction{sr}\AgdaSpace{}%
\AgdaBound{fc}\AgdaSpace{}%
\AgdaInductiveConstructor{ε}\AgdaSpace{}%
\AgdaSymbol{=}\AgdaSpace{}%
\AgdaBound{fc}\<%
\\
\>[0]\AgdaFunction{sr}\AgdaSpace{}%
\AgdaBound{fc}\AgdaSpace{}%
\AgdaSymbol{(\AgdaUnderscore{}}\AgdaSpace{}%
\AgdaOperator{\AgdaInductiveConstructor{◅}}\AgdaSpace{}%
\AgdaSymbol{\AgdaUnderscore{})}\AgdaSpace{}%
\AgdaKeyword{with}\AgdaSpace{}%
\AgdaBound{fc}\AgdaSpace{}%
\AgdaSymbol{.}\AgdaField{CoInd⟦\AgdaUnderscore{}⟧.unfold}\<%
\\
\>[0]\AgdaFunction{sr}\AgdaSpace{}%
\AgdaSymbol{\AgdaUnderscore{}}\AgdaSpace{}%
\AgdaSymbol{(}\AgdaInductiveConstructor{sync}\AgdaSpace{}%
\AgdaSymbol{(}\AgdaInductiveConstructor{out}\AgdaSpace{}%
\AgdaBound{fx}\AgdaSymbol{)}\AgdaSpace{}%
\AgdaInductiveConstructor{inp}\AgdaSpace{}%
\AgdaOperator{\AgdaInductiveConstructor{◅}}\AgdaSpace{}%
\AgdaSymbol{\AgdaUnderscore{})}\AgdaSpace{}%
\AgdaSymbol{|}\AgdaSpace{}%
\AgdaInductiveConstructor{success}\AgdaSpace{}%
\AgdaOperator{\AgdaInductiveConstructor{,}}\AgdaSpace{}%
\AgdaSymbol{((\AgdaUnderscore{}}\AgdaSpace{}%
\AgdaOperator{\AgdaInductiveConstructor{,}}\AgdaSpace{}%
\AgdaInductiveConstructor{success}\AgdaSpace{}%
\AgdaSymbol{(}\AgdaInductiveConstructor{out}\AgdaSpace{}%
\AgdaBound{e}\AgdaSymbol{)}\AgdaSpace{}%
\AgdaSymbol{\AgdaUnderscore{})}\AgdaSpace{}%
\AgdaOperator{\AgdaInductiveConstructor{,}}\AgdaSpace{}%
\AgdaSymbol{\AgdaUnderscore{})}\AgdaSpace{}%
\AgdaOperator{\AgdaInductiveConstructor{,}}\AgdaSpace{}%
\AgdaInductiveConstructor{refl}\AgdaSpace{}%
\AgdaOperator{\AgdaInductiveConstructor{,}}\AgdaSpace{}%
\AgdaSymbol{\AgdaUnderscore{}}\AgdaSpace{}%
\AgdaSymbol{=}\<%
\\
\>[0][@{}l@{\AgdaIndent{0}}]%
\>[2]\AgdaFunction{⊥-elim}\AgdaSpace{}%
\AgdaSymbol{(}\AgdaBound{e}\AgdaSpace{}%
\AgdaSymbol{\AgdaUnderscore{}}\AgdaSpace{}%
\AgdaBound{fx}\AgdaSymbol{)}\<%
\\
\>[0]\AgdaFunction{sr}\AgdaSpace{}%
\AgdaSymbol{\AgdaUnderscore{}}\AgdaSpace{}%
\AgdaSymbol{(}\AgdaInductiveConstructor{sync}\AgdaSpace{}%
\AgdaInductiveConstructor{inp}\AgdaSpace{}%
\AgdaSymbol{(}\AgdaInductiveConstructor{out}\AgdaSpace{}%
\AgdaBound{gx}\AgdaSymbol{)}\AgdaSpace{}%
\AgdaOperator{\AgdaInductiveConstructor{◅}}\AgdaSpace{}%
\AgdaBound{reds}\AgdaSymbol{)}\AgdaSpace{}%
\AgdaSymbol{|}\AgdaSpace{}%
\AgdaInductiveConstructor{inp-out}\AgdaSpace{}%
\AgdaOperator{\AgdaInductiveConstructor{,}}\AgdaSpace{}%
\AgdaSymbol{\AgdaUnderscore{}}\AgdaSpace{}%
\AgdaOperator{\AgdaInductiveConstructor{,}}\AgdaSpace{}%
\AgdaInductiveConstructor{refl}\AgdaSpace{}%
\AgdaOperator{\AgdaInductiveConstructor{,}}\AgdaSpace{}%
\AgdaBound{pr}\AgdaSpace{}%
\AgdaSymbol{=}\AgdaSpace{}%
\AgdaFunction{sr}\AgdaSpace{}%
\AgdaSymbol{(}\AgdaBound{pr}\AgdaSpace{}%
\AgdaSymbol{(\AgdaUnderscore{}}\AgdaSpace{}%
\AgdaOperator{\AgdaInductiveConstructor{,}}\AgdaSpace{}%
\AgdaBound{gx}\AgdaSymbol{))}\AgdaSpace{}%
\AgdaBound{reds}\<%
\\
\>[0]\AgdaFunction{sr}\AgdaSpace{}%
\AgdaSymbol{\AgdaUnderscore{}}\AgdaSpace{}%
\AgdaSymbol{(}\AgdaInductiveConstructor{sync}\AgdaSpace{}%
\AgdaSymbol{(}\AgdaInductiveConstructor{out}\AgdaSpace{}%
\AgdaBound{fx}\AgdaSymbol{)}\AgdaSpace{}%
\AgdaInductiveConstructor{inp}\AgdaSpace{}%
\AgdaOperator{\AgdaInductiveConstructor{◅}}\AgdaSpace{}%
\AgdaBound{reds}\AgdaSymbol{)}\AgdaSpace{}%
\AgdaSymbol{|}\AgdaSpace{}%
\AgdaInductiveConstructor{out-inp}\AgdaSpace{}%
\AgdaOperator{\AgdaInductiveConstructor{,}}\AgdaSpace{}%
\AgdaSymbol{\AgdaUnderscore{}}\AgdaSpace{}%
\AgdaOperator{\AgdaInductiveConstructor{,}}\AgdaSpace{}%
\AgdaInductiveConstructor{refl}\AgdaSpace{}%
\AgdaOperator{\AgdaInductiveConstructor{,}}\AgdaSpace{}%
\AgdaBound{pr}\AgdaSpace{}%
\AgdaSymbol{=}\AgdaSpace{}%
\AgdaFunction{sr}\AgdaSpace{}%
\AgdaSymbol{(}\AgdaBound{pr}\AgdaSpace{}%
\AgdaSymbol{(\AgdaUnderscore{}}\AgdaSpace{}%
\AgdaOperator{\AgdaInductiveConstructor{,}}\AgdaSpace{}%
\AgdaBound{fx}\AgdaSymbol{))}\AgdaSpace{}%
\AgdaBound{reds}\<%
\end{code}

Note that we have an absurd case, in which a successful configuration apparently reduces, which we rule out using false elimination ($\AgdaFunction{⊥-elim}$).
The soundness proof is a simple combination of the above auxiliary results. We
use the library function
\[
  \AgdaFunction{fcoind-to-ind} :
  \forall\{\AgdaBound{is} : \AgdaDatatype{IS}~\AgdaDatatype{U}\}
         \{\AgdaBound{cois} : \AgdaDatatype{IS}~\AgdaDatatype{U}\}
  ~\AgdaSymbol{→}~
  \AgdaFunction{FCoInd⟦}~\AgdaBound{is}~\AgdaFunction{,}~\AgdaBound{cois}~\AgdaFunction{⟧}
  ~\AgdaFunction{⊆}~
  \AgdaFunction{Ind⟦}~\AgdaBound{is}~\AgdaFunction{∪}~\AgdaBound{cois}~\AgdaFunction{⟧}
\]
to extract an inductive derivation of $\AgdaFunction{FCompI}~\mathit{Se}$ from a
derivation of $\AgdaFunction{FCompG}~\mathit{Se}$ in the \GIS for fair
compliance.
\begin{code}%
\>[0]\AgdaFunction{sound}\AgdaSpace{}%
\AgdaSymbol{:}\AgdaSpace{}%
\AgdaSymbol{∀\{}\AgdaBound{Se}\AgdaSymbol{\}}\AgdaSpace{}%
\AgdaSymbol{→}\AgdaSpace{}%
\AgdaFunction{FCompG}\AgdaSpace{}%
\AgdaBound{Se}\AgdaSpace{}%
\AgdaSymbol{→}\AgdaSpace{}%
\AgdaFunction{FCompS}\AgdaSpace{}%
\AgdaBound{Se}\<%
\\
\>[0]\AgdaFunction{sound}\AgdaSpace{}%
\AgdaBound{fc}\AgdaSpace{}%
\AgdaBound{reds}\AgdaSpace{}%
\AgdaSymbol{=}\AgdaSpace{}%
\AgdaFunction{FCompI→MaySucceed}\AgdaSpace{}%
\AgdaSymbol{(}\AgdaFunction{fcoind-to-ind}\AgdaSpace{}%
\AgdaSymbol{(}\AgdaFunction{sr}\AgdaSpace{}%
\AgdaBound{fc}\AgdaSpace{}%
\AgdaBound{reds}\AgdaSymbol{))}\<%
\end{code}

\subsection{Completeness}

For the completeness result we appeal to the bounded coinduction principle of \GISs (Proposition~\ref{prop:bcp}), which requires us to prove boundedness and consistency of $\AgdaFunction{FCompS}$. Concerning boundedness, we start by computing a proof of $\AgdaFunction{FCompI}~\mathit{Se}$ for every session $\mathit{Se}$ that may reduce a successful configuration, by induction on the reduction.
\begin{code}%
\>[0]\AgdaFunction{MaySucceed→FCompI}\AgdaSpace{}%
\AgdaSymbol{:}\AgdaSpace{}%
\AgdaSymbol{∀\{}\AgdaBound{Se}\AgdaSymbol{\}}\AgdaSpace{}%
\AgdaSymbol{→}\AgdaSpace{}%
\AgdaFunction{MaySucceed}\AgdaSpace{}%
\AgdaBound{Se}\AgdaSpace{}%
\AgdaSymbol{→}\AgdaSpace{}%
\AgdaFunction{FCompI}\AgdaSpace{}%
\AgdaBound{Se}\<%
\\
\>[0]\AgdaFunction{MaySucceed→FCompI}\AgdaSpace{}%
\AgdaSymbol{(\AgdaUnderscore{}}\AgdaSpace{}%
\AgdaOperator{\AgdaInductiveConstructor{,}}\AgdaSpace{}%
\AgdaBound{reds}\AgdaSpace{}%
\AgdaOperator{\AgdaInductiveConstructor{,}}\AgdaSpace{}%
\AgdaBound{succ}\AgdaSymbol{)}\AgdaSpace{}%
\AgdaSymbol{=}\AgdaSpace{}%
\AgdaFunction{aux}\AgdaSpace{}%
\AgdaBound{reds}\AgdaSpace{}%
\AgdaBound{succ}\<%
\\
\>[0][@{}l@{\AgdaIndent{0}}]%
\>[2]\AgdaKeyword{where}\<%
\\
\>[2][@{}l@{\AgdaIndent{0}}]%
\>[4]\AgdaFunction{aux}\AgdaSpace{}%
\AgdaSymbol{:}\AgdaSpace{}%
\AgdaSymbol{∀\{}\AgdaBound{Se}\AgdaSpace{}%
\AgdaBound{Se'}\AgdaSymbol{\}}\AgdaSpace{}%
\AgdaSymbol{→}\AgdaSpace{}%
\AgdaFunction{Reductions}\AgdaSpace{}%
\AgdaBound{Se}\AgdaSpace{}%
\AgdaBound{Se'}\AgdaSpace{}%
\AgdaSymbol{→}\AgdaSpace{}%
\AgdaDatatype{Success}\AgdaSpace{}%
\AgdaBound{Se'}\AgdaSpace{}%
\AgdaSymbol{→}\AgdaSpace{}%
\AgdaFunction{FCompI}\AgdaSpace{}%
\AgdaBound{Se}\<%
\\
\>[4]\AgdaFunction{aux}\AgdaSpace{}%
\AgdaInductiveConstructor{ε}\AgdaSpace{}%
\AgdaBound{succ}\AgdaSpace{}%
\AgdaSymbol{=}\AgdaSpace{}%
\AgdaFunction{apply-ind}\AgdaSpace{}%
\AgdaSymbol{(}\AgdaInductiveConstructor{inj₁}\AgdaSpace{}%
\AgdaInductiveConstructor{success}\AgdaSymbol{)}\AgdaSpace{}%
\AgdaSymbol{(\AgdaUnderscore{}}\AgdaSpace{}%
\AgdaOperator{\AgdaInductiveConstructor{,}}\AgdaSpace{}%
\AgdaBound{succ}\AgdaSymbol{)}\AgdaSpace{}%
\AgdaSymbol{λ}\AgdaSpace{}%
\AgdaSymbol{()}\<%
\\
\>[4]\AgdaFunction{aux}\AgdaSpace{}%
\AgdaSymbol{(}\AgdaInductiveConstructor{sync}\AgdaSpace{}%
\AgdaSymbol{(}\AgdaInductiveConstructor{out}\AgdaSpace{}%
\AgdaBound{fx}\AgdaSymbol{)}\AgdaSpace{}%
\AgdaInductiveConstructor{inp}\AgdaSpace{}%
\AgdaOperator{\AgdaInductiveConstructor{◅}}\AgdaSpace{}%
\AgdaBound{red}\AgdaSymbol{)}\AgdaSpace{}%
\AgdaBound{succ}\AgdaSpace{}%
\AgdaSymbol{=}\<%
\\
\>[4][@{}l@{\AgdaIndent{0}}]%
\>[6]\AgdaFunction{apply-ind}\AgdaSpace{}%
\AgdaSymbol{(}\AgdaInductiveConstructor{inj₂}\AgdaSpace{}%
\AgdaInductiveConstructor{out-inp}\AgdaSymbol{)}\AgdaSpace{}%
\AgdaSymbol{(\AgdaUnderscore{}}\AgdaSpace{}%
\AgdaOperator{\AgdaInductiveConstructor{,}}\AgdaSpace{}%
\AgdaSymbol{\AgdaUnderscore{}}\AgdaSpace{}%
\AgdaOperator{\AgdaInductiveConstructor{,}}\AgdaSpace{}%
\AgdaBound{fx}\AgdaSymbol{)}\AgdaSpace{}%
\AgdaSymbol{λ\{}\AgdaInductiveConstructor{Data.Fin.zero}\AgdaSpace{}%
\AgdaSymbol{→}\AgdaSpace{}%
\AgdaFunction{aux}\AgdaSpace{}%
\AgdaBound{red}\AgdaSpace{}%
\AgdaBound{succ}\AgdaSymbol{\}}\<%
\\
\>[4]\AgdaFunction{aux}\AgdaSpace{}%
\AgdaSymbol{(}\AgdaInductiveConstructor{sync}\AgdaSpace{}%
\AgdaInductiveConstructor{inp}\AgdaSpace{}%
\AgdaSymbol{(}\AgdaInductiveConstructor{out}\AgdaSpace{}%
\AgdaBound{gx}\AgdaSymbol{)}\AgdaSpace{}%
\AgdaOperator{\AgdaInductiveConstructor{◅}}\AgdaSpace{}%
\AgdaBound{red}\AgdaSymbol{)}\AgdaSpace{}%
\AgdaBound{succ}\AgdaSpace{}%
\AgdaSymbol{=}\<%
\\
\>[4][@{}l@{\AgdaIndent{0}}]%
\>[6]\AgdaFunction{apply-ind}\AgdaSpace{}%
\AgdaSymbol{(}\AgdaInductiveConstructor{inj₂}\AgdaSpace{}%
\AgdaInductiveConstructor{inp-out}\AgdaSymbol{)}\AgdaSpace{}%
\AgdaSymbol{(\AgdaUnderscore{}}\AgdaSpace{}%
\AgdaOperator{\AgdaInductiveConstructor{,}}\AgdaSpace{}%
\AgdaSymbol{\AgdaUnderscore{}}\AgdaSpace{}%
\AgdaOperator{\AgdaInductiveConstructor{,}}\AgdaSpace{}%
\AgdaBound{gx}\AgdaSymbol{)}\AgdaSpace{}%
\AgdaSymbol{λ\{}\AgdaInductiveConstructor{Data.Fin.zero}\AgdaSpace{}%
\AgdaSymbol{→}\AgdaSpace{}%
\AgdaFunction{aux}\AgdaSpace{}%
\AgdaBound{red}\AgdaSpace{}%
\AgdaBound{succ}\AgdaSymbol{\}}\<%
\end{code}

Then, boundedness follows by observing that $R$ fairly compliant with $S$ implies the existence of a successful configuration reachable from $\session{R}{S}$.
\begin{code}%
\>[0]\AgdaFunction{bounded}\AgdaSpace{}%
\AgdaSymbol{:}\AgdaSpace{}%
\AgdaSymbol{∀\{}\AgdaBound{Se}\AgdaSymbol{\}}\AgdaSpace{}%
\AgdaSymbol{→}\AgdaSpace{}%
\AgdaFunction{FCompS}\AgdaSpace{}%
\AgdaBound{Se}\AgdaSpace{}%
\AgdaSymbol{→}\AgdaSpace{}%
\AgdaFunction{FCompI}\AgdaSpace{}%
\AgdaBound{Se}\<%
\\
\>[0]\AgdaFunction{bounded}\AgdaSpace{}%
\AgdaBound{fc}\AgdaSpace{}%
\AgdaSymbol{=}\AgdaSpace{}%
\AgdaFunction{MaySucceed→FCompI}\AgdaSpace{}%
\AgdaSymbol{(}\AgdaBound{fc}\AgdaSpace{}%
\AgdaInductiveConstructor{ε}\AgdaSymbol{)}\<%
\end{code}

Showing that $\AgdaFunction{FCompS}$ is consistent means showing that every configuration $\AgdaBound{Se}$ that satisfies $\AgdaFunction{FCompS}$ is found in the conclusion of a rule in the inference system $\AgdaFunction{FCompIS}$ whose premises are all configurations that in turn satisfy $\AgdaFunction{FCompS}$. This follows by a straightforward case analysis on the \emph{first} reduction of $\AgdaBound{Se}$ that leads to a successful configuration.
\begin{code}%
\>[0]\AgdaFunction{consistent}\AgdaSpace{}%
\AgdaSymbol{:}\AgdaSpace{}%
\AgdaSymbol{∀\{}\AgdaBound{Se}\AgdaSymbol{\}}\AgdaSpace{}%
\AgdaSymbol{→}\AgdaSpace{}%
\AgdaFunction{FCompS}\AgdaSpace{}%
\AgdaBound{Se}\AgdaSpace{}%
\AgdaSymbol{→}\AgdaSpace{}%
\AgdaOperator{\AgdaFunction{ISF[}}\AgdaSpace{}%
\AgdaFunction{FCompIS}\AgdaSpace{}%
\AgdaOperator{\AgdaFunction{]}}\AgdaSpace{}%
\AgdaFunction{FCompS}\AgdaSpace{}%
\AgdaBound{Se}\<%
\\
\>[0]\AgdaFunction{consistent}\AgdaSpace{}%
\AgdaBound{fc}\AgdaSpace{}%
\AgdaKeyword{with}\AgdaSpace{}%
\AgdaBound{fc}\AgdaSpace{}%
\AgdaInductiveConstructor{ε}\<%
\\
\>[0]\AgdaSymbol{...}\AgdaSpace{}%
\AgdaSymbol{|}\AgdaSpace{}%
\AgdaSymbol{\AgdaUnderscore{}}\AgdaSpace{}%
\AgdaOperator{\AgdaInductiveConstructor{,}}\AgdaSpace{}%
\AgdaInductiveConstructor{ε}\AgdaSpace{}%
\AgdaOperator{\AgdaInductiveConstructor{,}}\AgdaSpace{}%
\AgdaBound{succ}\AgdaSpace{}%
\AgdaSymbol{=}\AgdaSpace{}%
\AgdaInductiveConstructor{success}\AgdaSpace{}%
\AgdaOperator{\AgdaInductiveConstructor{,}}\AgdaSpace{}%
\AgdaSymbol{(\AgdaUnderscore{}}\AgdaSpace{}%
\AgdaOperator{\AgdaInductiveConstructor{,}}\AgdaSpace{}%
\AgdaBound{succ}\AgdaSymbol{)}\AgdaSpace{}%
\AgdaOperator{\AgdaInductiveConstructor{,}}\AgdaSpace{}%
\AgdaInductiveConstructor{refl}\AgdaSpace{}%
\AgdaOperator{\AgdaInductiveConstructor{,}}\AgdaSpace{}%
\AgdaSymbol{λ}\AgdaSpace{}%
\AgdaSymbol{()}\<%
\\
\>[0]\AgdaSymbol{...}\AgdaSpace{}%
\AgdaSymbol{|}\AgdaSpace{}%
\AgdaSymbol{\AgdaUnderscore{}}\AgdaSpace{}%
\AgdaOperator{\AgdaInductiveConstructor{,}}\AgdaSpace{}%
\AgdaInductiveConstructor{sync}\AgdaSpace{}%
\AgdaSymbol{(}\AgdaInductiveConstructor{out}\AgdaSpace{}%
\AgdaBound{fx}\AgdaSymbol{)}\AgdaSpace{}%
\AgdaInductiveConstructor{inp}\AgdaSpace{}%
\AgdaOperator{\AgdaInductiveConstructor{◅}}\AgdaSpace{}%
\AgdaSymbol{\AgdaUnderscore{}}\AgdaSpace{}%
\AgdaOperator{\AgdaInductiveConstructor{,}}\AgdaSpace{}%
\AgdaSymbol{\AgdaUnderscore{}}\AgdaSpace{}%
\AgdaSymbol{=}\<%
\\
\>[0][@{}l@{\AgdaIndent{0}}]%
\>[2]\AgdaInductiveConstructor{out-inp}\AgdaSpace{}%
\AgdaOperator{\AgdaInductiveConstructor{,}}\AgdaSpace{}%
\AgdaSymbol{(\AgdaUnderscore{}}\AgdaSpace{}%
\AgdaOperator{\AgdaInductiveConstructor{,}}\AgdaSpace{}%
\AgdaSymbol{\AgdaUnderscore{}}\AgdaSpace{}%
\AgdaOperator{\AgdaInductiveConstructor{,}}\AgdaSpace{}%
\AgdaBound{fx}\AgdaSymbol{)}\AgdaSpace{}%
\AgdaOperator{\AgdaInductiveConstructor{,}}\AgdaSpace{}%
\AgdaInductiveConstructor{refl}\AgdaSpace{}%
\AgdaOperator{\AgdaInductiveConstructor{,}}\AgdaSpace{}%
\AgdaSymbol{λ}\AgdaSpace{}%
\AgdaSymbol{(\AgdaUnderscore{}}\AgdaSpace{}%
\AgdaOperator{\AgdaInductiveConstructor{,}}\AgdaSpace{}%
\AgdaBound{gx}\AgdaSymbol{)}\AgdaSpace{}%
\AgdaBound{reds}\AgdaSpace{}%
\AgdaSymbol{→}\AgdaSpace{}%
\AgdaBound{fc}\AgdaSpace{}%
\AgdaSymbol{(}\AgdaInductiveConstructor{sync}\AgdaSpace{}%
\AgdaSymbol{(}\AgdaInductiveConstructor{out}\AgdaSpace{}%
\AgdaBound{gx}\AgdaSymbol{)}\AgdaSpace{}%
\AgdaInductiveConstructor{inp}\AgdaSpace{}%
\AgdaOperator{\AgdaInductiveConstructor{◅}}\AgdaSpace{}%
\AgdaBound{reds}\AgdaSymbol{)}\<%
\\
\>[0]\AgdaSymbol{...}\AgdaSpace{}%
\AgdaSymbol{|}\AgdaSpace{}%
\AgdaSymbol{\AgdaUnderscore{}}\AgdaSpace{}%
\AgdaOperator{\AgdaInductiveConstructor{,}}\AgdaSpace{}%
\AgdaInductiveConstructor{sync}\AgdaSpace{}%
\AgdaInductiveConstructor{inp}\AgdaSpace{}%
\AgdaSymbol{(}\AgdaInductiveConstructor{out}\AgdaSpace{}%
\AgdaBound{gx}\AgdaSymbol{)}\AgdaSpace{}%
\AgdaOperator{\AgdaInductiveConstructor{◅}}\AgdaSpace{}%
\AgdaSymbol{\AgdaUnderscore{}}\AgdaSpace{}%
\AgdaOperator{\AgdaInductiveConstructor{,}}\AgdaSpace{}%
\AgdaSymbol{\AgdaUnderscore{}}\AgdaSpace{}%
\AgdaSymbol{=}\<%
\\
\>[0][@{}l@{\AgdaIndent{0}}]%
\>[2]\AgdaInductiveConstructor{inp-out}\AgdaSpace{}%
\AgdaOperator{\AgdaInductiveConstructor{,}}\AgdaSpace{}%
\AgdaSymbol{(\AgdaUnderscore{}}\AgdaSpace{}%
\AgdaOperator{\AgdaInductiveConstructor{,}}\AgdaSpace{}%
\AgdaSymbol{\AgdaUnderscore{}}\AgdaSpace{}%
\AgdaOperator{\AgdaInductiveConstructor{,}}\AgdaSpace{}%
\AgdaBound{gx}\AgdaSymbol{)}\AgdaSpace{}%
\AgdaOperator{\AgdaInductiveConstructor{,}}\AgdaSpace{}%
\AgdaInductiveConstructor{refl}\AgdaSpace{}%
\AgdaOperator{\AgdaInductiveConstructor{,}}\AgdaSpace{}%
\AgdaSymbol{λ}\AgdaSpace{}%
\AgdaSymbol{(\AgdaUnderscore{}}\AgdaSpace{}%
\AgdaOperator{\AgdaInductiveConstructor{,}}\AgdaSpace{}%
\AgdaBound{fx}\AgdaSymbol{)}\AgdaSpace{}%
\AgdaBound{reds}\AgdaSpace{}%
\AgdaSymbol{→}\AgdaSpace{}%
\AgdaBound{fc}\AgdaSpace{}%
\AgdaSymbol{(}\AgdaInductiveConstructor{sync}\AgdaSpace{}%
\AgdaInductiveConstructor{inp}\AgdaSpace{}%
\AgdaSymbol{(}\AgdaInductiveConstructor{out}\AgdaSpace{}%
\AgdaBound{fx}\AgdaSymbol{)}\AgdaSpace{}%
\AgdaOperator{\AgdaInductiveConstructor{◅}}\AgdaSpace{}%
\AgdaBound{reds}\AgdaSymbol{)}\<%
\end{code}

We obtain the completeness proof using the library function
\begin{align*}
  \AgdaFunction{bounded-coind} & ~\AgdaSymbol{:}~ 
  \AgdaSymbol{(}\AgdaBound{is}~\AgdaSymbol{:}~\AgdaDatatype{IS}~\AgdaDatatype{U}\AgdaSymbol{)}
  \AgdaSymbol(\AgdaBound{cois}~\AgdaSymbol:~\AgdaDatatype{IS}~\AgdaDatatype{U}\AgdaSymbol)
  \AgdaSymbol(\AgdaBound{spec}~\AgdaSymbol:~\AgdaDatatype{U}~\AgdaSymbol{→}~\AgdaDatatype{Set}\AgdaSymbol)
  \\
  & ~\AgdaSymbol{→}~
  \AgdaBound{spec}~\AgdaFunction{⊆}~\AgdaFunction{Ind⟦}~\AgdaBound{is}~\AgdaFunction{∪}~\AgdaBound{cois}~\AgdaFunction{⟧}
  ~\AgdaSymbol{→}~
  \AgdaBound{spec}~\AgdaFunction{⊆}~\AgdaFunction{ISF[}~\AgdaBound{is}~\AgdaFunction{]}~\AgdaBound{spec}
  \\
  & ~\AgdaSymbol{→}~
  \AgdaBound{spec}~\AgdaFunction{⊆}~\AgdaFunction{FCoInd⟦}~\AgdaBound{is}~\AgdaFunction{,}~\AgdaBound{cois}~\AgdaFunction{⟧}
\end{align*}
which applies the bounded coinduction principle to the boundedness and
consistency proofs.
\begin{code}%
\>[0]\AgdaFunction{complete}\AgdaSpace{}%
\AgdaSymbol{:}\AgdaSpace{}%
\AgdaSymbol{∀\{}\AgdaBound{Se}\AgdaSymbol{\}}\AgdaSpace{}%
\AgdaSymbol{→}\AgdaSpace{}%
\AgdaFunction{FCompS}\AgdaSpace{}%
\AgdaBound{Se}\AgdaSpace{}%
\AgdaSymbol{→}\AgdaSpace{}%
\AgdaFunction{FCompG}\AgdaSpace{}%
\AgdaBound{Se}\<%
\\
\>[0]\AgdaFunction{complete}\AgdaSpace{}%
\AgdaSymbol{=}\AgdaSpace{}%
\AgdaOperator{\AgdaFunction{bounded-coind[}}\AgdaSpace{}%
\AgdaFunction{FCompIS}\AgdaSpace{}%
\AgdaOperator{\AgdaFunction{,}}\AgdaSpace{}%
\AgdaFunction{FCompCOIS}\AgdaSpace{}%
\AgdaOperator{\AgdaFunction{]}}\AgdaSpace{}%
\AgdaFunction{FCompS}\AgdaSpace{}%
\AgdaFunction{bounded}\AgdaSpace{}%
\AgdaFunction{consistent}\<%
\end{code}


%% file: conclusion.tex
\section{Concluding Remarks}
\label{sec:conclusion}

We have shown that generalized inference systems are an effective framework for
defining sound and complete proof systems of (some) mixed safety and liveness
properties of (dependent) session types (Definitions~\ref{def:wt} and~\ref{def:fcomp}), as well as of a liveness-preserving subtyping relation
(Definition~\ref{def:fsub}).
We think that this achievement is more than a coincidence.  One of the
fundamental results in model checking states that every property can be
expressed as the conjunction of a safety property and a liveness
property~\cite{AlpernSchneider85,AlpernSchneider87,BaierKatoen08}. The
connections between safety and liveness on one side and coinduction and
induction on the other make \GISs appropriate for characterizing combined safety
and liveness properties.

Murgia~\cite{Murgia19} studies a variety of compliance relations for processes
and session types, showing that many of them are fixed points of a functional
operator, but not necessarily the least or the greatest ones. In particular, he
shows that \emph{progress compliance}, which is akin to our compliance
(Definition~\ref{def:comp}), is a greatest fixed point and that
\emph{should-testing compliance}, which is akin to our fair compliance
(Definition~\ref{def:fcomp}), is an intermediate fixed point. These results are
consistent with Theorem~\ref{thm:compliance}.
We have extended these results to subtyping (Definiton~\ref{def:sub}) and fair
subtyping (Definition~\ref{def:fsub}). Previous alternative characterizations of
fair subtyping and the related \emph{should-testing preorder} either require the
use of auxiliary relations~\cite{Padovani13,Padovani16} or are denotational in
nature~\cite{RensinkVogler07} and therefore not as insightful as desirable.
Using \GISs, we have obtained complete characterizations of fair compliance and
fair subtyping by simply \emph{adding a few corules} to the proof systems of
their ``unfair'' counterparts.

We have coded all the notions and results discussed in the paper in
Agda~\cite{Norell07}, thus providing the first machine-checked formalization of
liveness properties and liveness-preserving subtyping relations for dependent
session types.
The Agda formalization is not entirely constructive since it makes use of three
postulates: the \emph{extensionality axiom}, the \emph{law of excluded middle},
and a specific instance of this latter postulate concerning the inductive
characterization of \emph{convergence} (see~\refrule{s-converge} in
Table~\ref{tab:subt}) and its negation, which is characterized using an
existentially quantified, coinductive definition.
Note that the Agda library for \GISs is a standalone
development~\cite{Ciccone20,CicconeDagninoZucca21,CicconeDagninoZucca21B}, on
top of which we have built our own~\cite{CicconePadovani21}. This makes it easy
to extend our results to other families of processes or to different properties.
Other mechanizations (in Coq) of session types have been presented by
Castro-Perez et al.~\cite{Castro-PerezFGY21} and by Cruz-Felipe et al.~\cite{Cruz-FilipeMP21}.
In both cases, types are represented as an inductive datatype and the usual
recursion operator, but Castro-Perez et al.~\cite{Castro-PerezFGY21} also
provide a representation based on coinductive trees that they prove to be
trace-equivalent to the recursive one and that is similar to our own
(Section~\ref{sec:agda}).
Cruz-Felipe et al.~\cite{Cruz-FilipeMP21} restrict labels to a two-value set. We
have made the same design choice in the original version of this paper~\cite{CicconePadovani21B}, while in the present one we have generalized the
formalization to arbitrary sets with decidable equality.

In this paper we have focused on properties of session types alone. In a
different work~\cite{CicconeP22}, we have successfully integrated the \GIS for
fair subtyping (Section~\ref{sec:subtyping}) into a session type system for the
enforcement of liveness properties of processes. This problem has remained open
for a long time~\cite{Padovani13,Padovani16} because the integration of fair
subtyping into a coinductively-interpreted session type system is
(unsurprisingly) challenging. By contrast, session type systems making use of
safety-preserving subtyping relations are quite
widespread~\cite{GayHole05,CastagnaDezaniGiachinoPadovani09,HuttelEtAl16,CastagnaDezaniGiachinoPadovani19}.
In accordance with the achievements described in this paper, \GISs proved to be
appropriate for defining such type system.
A natural development of the results presented in this paper is their extension
to multiparty and higher-order session types. Both extensions appear to be
technically easy to accommodate. In particular, the characterization of fair
subtyping for multiparty session types is essentially the same as that for the
binary case~\cite{Padovani16}. Also, Ciccone and Padovani~\cite{CicconeP22} have
shown that uncontrolled variance of higher-order session types may break the
liveness-preserving feature of fair subtyping. To which extent fair subtyping
can be relaxed to allow for co/contra-variance of higher-order session types
remains to be established.
